\documentclass[journal, draftclsnofoot, onecolumn, romanappendices]{IEEEtran}

	\usepackage[T1]{fontenc}		
	\usepackage[utf8]{inputenc}	
	\usepackage[english]{babel}	

	\usepackage{url}
	\usepackage{ifthen}

	\usepackage{amsthm,amssymb}
	\usepackage[cmex10]{amsmath} 
	\usepackage{mathtools}
	\usepackage{bm}
	\usepackage{nicefrac}
	\usepackage{dsfont} 
	\usepackage[subnum]{cases}  

	\usepackage{subfiles} 

	\usepackage{placeins} 	
	\usepackage{caption}
	\usepackage{subcaption}
	\usepackage{tikz}
		\usetikzlibrary{shapes, 
										arrows,
										fit,
										shapes,
										snakes,
										positioning,
										decorations.markings,
										calc}	
	\usepackage{pgfplots}
		\pgfplotsset{compat=1.11}
		\usepgfplotslibrary{fillbetween}
	\usepackage[percent]{overpic}
	\usepackage{pict2e} 
	\usepackage[most]{tcolorbox}

	\usepackage{array}			
	
	\usepackage[final]{microtype} 
	\usepackage[subtle,mathspacing=normal,mathdisplays=normal]{savetrees} %
	
	\usepackage{color}

	\usepackage{hyperref}	
	\usepackage{cite}

	\newtheorem{theorem}{Theorem}
	\newtheorem{lemma}{Lemma}
	\newtheorem{corollary}{Corollary}
	\newtheorem{proposition}{Proposition}

	\newtheorem{definition}{Definition}
	
	\newtheorem{example}{Example}
	
	\newtheorem*{assumption*}{Assumption}
	
	\newtheorem{remark}{Remark}

	\allowdisplaybreaks  
	\makeatletter 
	\newcommand{\leqnomode}{\tagsleft@true}
	\newcommand{\reqnomode}{\tagsleft@false}
	\makeatother

	\newtoggle{draft}
	\toggletrue{draft}

\tikzset{
	barrow/.style={
		decoration={markings,mark=at position 1 with {\arrow[scale=1.5,#1]{>}}},
		postaction={decorate},
		shorten >=0.4pt},
	barrow/.default=black}

\newcommand{\bHa}{\bH_{\textnormal{A}}}
\newcommand{\bHp}{\bH_{\textnormal{P}}}

	\definecolor{mycolor1}{rgb}{0.00000,0.44700,0.74100}%
	\definecolor{mycolor2}{rgb}{0.85000,0.32500,0.09800}%
	\definecolor{mycolor3}{rgb}{0.92900,0.69400,0.12500}%
	\definecolor{mycolor4}{rgb}{0.49400,0.18400,0.55600}%
	\definecolor{mycolor5}{rgb}{0.300000,0.64314,0.00000}%
	\colorlet{mygreen}{green!60!black}%
	\colorlet{myblue}{mycolor1!80!black}%
	\definecolor{gray}{rgb}{0.6,0.6,0.6}
		
\newcommand\itemEq[1][]{%
  \ifx\relax#1\relax  \item \else \item[#1] \fi
  \abovedisplayskip=0pt\abovedisplayshortskip=0pt~\vspace*{-\baselineskip}}
	
\newcommand{\dd}{\mathop{}\!\mathrm{d}}	

\newcommand{\myoverset}[2]{\ensuremath{\overset{\mathclap{#1}}{#2}}}

\newcommand{\itb}{\begin{itemize}}
\newcommand{\ite}{\end{itemize}}
\newcommand{\enb}{\begin{enumerate}}
\newcommand{\ene}{\end{enumerate}}
\newcommand{\eqm}[1]{\begin{align}#1\end{align}}

\newcommand*{\dotleq}{\mathrel{\dot{\leq}}}
\newcommand*{\dotgeq}{\mathrel{\dot{\geq}}}

\newcommand{\Trans}{{\mathrm{T}}}

\newcommand{\LB}{\left(}
\newcommand{\RB}{\right)}
\newcommand{\LSB}{\left[}
\newcommand{\RSB}{\right]}

\newcommand{\abs}[1]{\left|#1\right|}

\newcommand{\bHH}{{{\hat{\bH}}}}

\newcommand{\expj}{^{(j)}}
\newcommand{\expo}{^{(1)}}
\newcommand{\expt}{^{(2)}}
\newcommand{\expK}{^{(K)}}

\newcommand{\Pb}{{\bar{P}}}

\newcommand{\Gck}{\overline{\Gc}_{\Kc}}
\newcommand{\gck}{\overline{g}_{\Kc}}

\newcommand\mhyp{{\hbox{-}}}

\DeclareMathAlphabet{\mathbit}{OML}{cmr}{bx}{it}
\DeclareMathAlphabet{\mathsf}{OT1}{cmss}{m}{n}
\DeclareMathAlphabet{\mathTXf}{OT1}{cmss}{bx}{it}

\DeclareMathOperator{\Transpose}{T}

\DeclareMathOperator{\Exp}{{\mathbb{E}}}

\DeclareMathOperator{\CN}{\mathcal{N}_{\mathbb{C}}}

\DeclareMathOperator{\DCSI}{DCSI}
\newcommand{\Weak}{\mathrm{Weak}} 

\DeclareMathOperator{\DoF}{DoF}

\newcommand{\BC}{{\text{BC}}}
\newcommand{\CCSI}{{\mathrm{CCSI}}}

\newcommand{\APZF}{\mathrm{APZF}}

\DeclareMathOperator{\SNR}{SNR}

\newcommand{\Rb}{{{\mathbb{R}}}}
\newcommand{\Cb}{{{\mathbb{C}}}}
\newcommand{\Wb}{{{\mathbb{W}}}}

\newcommand{\hv}{\mathbf{h}}

\newcommand{\xv}{\mathbf{x}}

\newcommand{\Cc}{{{\mathcal{C}}}}

\newcommand{\Gc}{{{\mathcal{G}}}}
\newcommand{\Hc}{{\mathcal{H}}}
\newcommand{\Ic}{{{\mathcal{I}}}}

\newcommand{\Kc}{{{\mathcal{K}}}}

\newcommand{\Sc}{{{\mathcal{S}}}}

\newcommand{\Uc}{{{\mathcal{U}}}}
\newcommand{\Vc}{{{\mathcal{V}}}}

\newcommand{\Xc}{{{\mathcal{X}}}}
\newcommand{\Yc}{{{\mathcal{Y}}}}

\newcommand{\bA}{\mathbf{A}}

\newcommand{\bF}{\mathbf{F}}

\newcommand{\bH}{\mathbf{H}}
\newcommand{\bI}{\mathbf{I}}

\newcommand{\bT}{\mathbf{T}}

\newcommand{\bd}{\bm{d}}
\newcommand{\be}{\bm{e}}

\newcommand{\bh}{\bm{h}}

\newcommand{\bs}{\bm{s}}

\newcommand{\by}{\bm{y}}

\newcommand{\bgamma}{{{\boldsymbol{\gamma}}}}

\newcommand{\gammabar}{{\bar{\gamma}}}

\newcommand{\norm}[1]{\lVert{#1}\rVert}

\newcommand{\Fro}{{\mathrm{F}}}

\newcommand{\E}{{\mathbb{E}}}

\newcommand{\trans}{{\text{T}}}

\newcommand{\He}{{{\mathrm{H}}}}

\newtcolorbox{whitebg}[1][]{
    arc=0mm,
    boxsep=0cm,
    toprule=0.0pt,
    leftrule=0.0pt,
    bottomrule=.0pt,
    rightrule=0.0pt,
    colframe=white,
		colback=white,
		top=3pt,
		bottom=3pt,
		hbox
}

\begin{document}
\title{On the Degrees-of-Freedom of the K-user Distributed Broadcast Channel}
\author{Antonio~Bazco-Nogueras,~\IEEEmembership{Member,~IEEE,}
        Paul de~Kerret,~\IEEEmembership{Member,~IEEE,}\\%
				David~Gesbert,~\IEEEmembership{Fellow,~IEEE,}
        and~Nicolas~Gresset,~\IEEEmembership{Senior~Member,~IEEE}
		\thanks{A. Bazco-Nogueras was with the Mitsubishi Electric Research and Development Centre Europe, 35708 Rennes, France, and also with the Communications Systems Department, Eurecom, 06410 Biot, France (e-mail:bazco@eurecom.fr).}
		\thanks{P. de Kerret and D. Gesbert are with the Communications Systems Department, Eurecom, 06410 Biot, France (e-mail: dekerret@eurecom.fr; gesbert@eurecom.fr). P. de Kerret and D. Gesbert also acknowledge the support of the ERC 670896.}
		\thanks{N. Gresset is with the Mitsubishi Electric Research and Development Centre Europe, 35708 Rennes, France (e-mail: n.gresset@fr.merce.mee.com).}
		\thanks{This work has been partially presented 
		at the 2016 IEEE International Symposium on Information Theory\cite{dekerret2016_ISIT}, and at the 2017 Colloque GRETSI\cite{Bazco2017_GRETSI}.}
}

\maketitle

		\begin{abstract}
					We study the Degrees-of-Freedom (DoF) in a wireless setting in which $K$ Transmitters (TXs) aim at jointly serving~$K$ users.  The performance is studied when the TXs are faced with a distributed Channel State Information (CSI) configuration in which each TX has access to its own multi-user imperfect channel estimate based on which it designs its transmit coefficients. The channel estimates are not only imperfectly acquired but they are also imperfectly shared between the TXs. Our first contribution consists of computing a genie-aided upper bound for the DoF of that setting. Our main contribution is then to develop a new robust transmission scheme that leverages the different qualities of CSI available at the TXs to improve the achieved DoF. We show the surprising result that there is a CSI regime, coined the Weak-CSIT regime, in which the genie-aided upper bound is achieved by the proposed transmission scheme. Interestingly, the optimal DoF in the Weak-CSIT regime only depends on the CSI quality at the best informed TX and not on the CSI quality at all other TXs.
		\end{abstract} 


		\section{Introduction}

			\subsection{Limited Channel State Information on the Transmitter Side in Wireless Networks}
			\IEEEPARstart{C}{hannel} capacity characterization of multi-user wireless networks is known to be an elusive problem for many practical scenarios, in particular for the cases in which the Channel State Information at the Transmitter (CSIT) is not perfect. 
			In order to tackle this problem, capacity approximations at high Signal-to-Noise Ratio (SNR), such as Degrees-of-Freedom (DoF) analysis \cite{Jafar2007}, have been used as a first step towards the complete characterization of the system capacity, and they have successfully led to many important insights. 
			For example, the DoF of the MISO Broadcast Channel (BC) with imperfect noisy CSIT was characterized by Davoodi and Jafar in \cite{Davoodi2016} by
			showing that a CSIT error variance scaling in $\SNR^{-\alpha}$, for $\alpha\in [0,1]$, leads to a DoF of~$1+(K-1)\alpha$.

			A different line of work in the area of BC with limited feedback has been focused on the exploitation of delayed CSIT. This research area was triggered by the seminal work from Maddah-Ali and Tse~\cite{MaddahAli2012} where it was shown that completely outdated CSIT could still be exploited via a multi-phase protocol involving the retransmission of the interference generated. While the original model assumed completely outdated CSIT, a large number of works have developed generalized schemes for the case of partially outdated \cite{Gou2012, Yang2013, dekerret2016_ITW}, alternating \cite{Tandon2012b}, or evolving CSIT \cite{Chen2013a}, to name just a~few.

			In all the above literature, however, \emph{centralized} CSIT is typically assumed, i.e., the transmission is optimized at the TX side on the basis of a \emph{single} imperfect/outdated channel estimate being common at every transmit antenna. 
			Recently, the increasing importance of cooperation of non-collocated TXs---as, for example, in Unmanned Aerial Device (UAV) aided networks \cite{Becvar2017}---has led to an increasing number of works challenging this assumption of centralized CSIT.  
			In \cite{Rao2013,dekerret2014_TWC}, methods have been developed to reduce the CSIT required to achieve MIMO Interference Alignment (IA), and the DoF achieved with delayed and local CSIT in the Interference Channel (IC) has been also studied in several works \cite{Hao2015,Lee2015,Vahid2016a}. 
			Furthermore, the assumption of centralized CSIT has been challenged in capacity analysis for the Multiple Access Channel \cite{Lapidoth2010} and the Relay Channel \cite{Kolte2016}, among others.

		\subsection{Network MIMO with Distributed CSIT Model}
			In order to account for TX-dependent limited feedback in the network MIMO channel, we focus in this work on a wireless configuration first studied in \cite{dekerret2012_TIT}, in which the user's data symbols are available and jointly transmitted from all TXs, whereas the channel estimates could only be imperfectly obtained at the TXs. Hereinafter, we refer to this setting as the Distributed BC setting\footnote{%
					In terms of DoF analysis, and under the assumption of centralized shared CSIT, a network MIMO setting in which several TXs jointly serve a set of users can be seen equivalent to the conventional BC setting. Moreover, we will use the BC setting as an ideal benchmark to which we compare the performance of the scenario considered. 
					We hence refer to the counterpart setting with different CSIT at each TX as the Distributed BC setting.}. 
			We will show below why such assumptions, although seemingly contradictory at first sight, are actually very relevant in current wireless networks, and even more in future 5G-and-beyond networks. 
			Yet, let us first try to convey the main intuition before diving in a more detailed and precise discussion: In many scenarios of interest, the latency constraint for data delivery is significantly looser than the CSI outdating constraint (which is related to the coherence time, and hence very short in many relevant mobility scenarios). 
			This has for consequence that the data caching or sharing between TXs can be achieved in practice while \emph{timely} CSI acquisition and sharing becomes the main bottleneck.
			
			\subsection*{Imperfect CSI Acquisition and Sharing}
				We consider in this work a TX-dependent limited feedback where each TX receives its own multi-user imperfect estimate. 
				This CSIT configuration is coined as the \emph{Distributed CSIT} configuration (D-CSIT), in opposition to the aforementioned centralized CSIT model. 
				The imperfect multi-user channel estimate is obtained at the TXs from a CSI acquisition and sharing mechanism not discussed in this work. 
				Yet, due to unavoidable delay and imperfections in the CSI sharing mechanism, this CSI sharing step leads to a setting where the TXs have received different imperfect estimates of the true channel. 
				We provide below several practical examples that are illustrated in Fig.~\ref{fig:examples}.
				
					\begin{figure}[t]%
						\begin{overpic}[scale=.48]{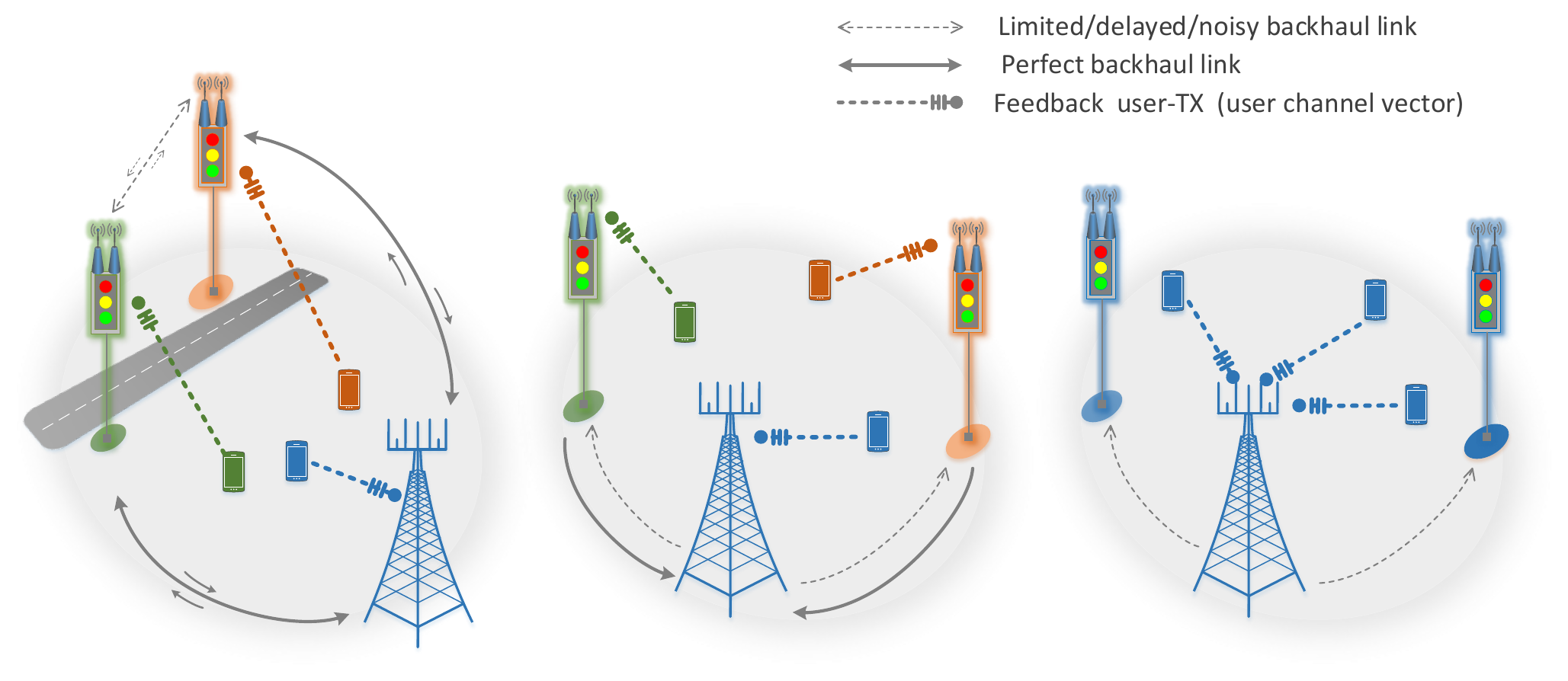}
							\put(18.25,27.5){\small $\hv_2$}
							\put(12.5,18){\small $\hv_3$}
							\put(-0.5,12.5){\small $\bH^{(3)}\!\!=\![\hv_1,\check{\hv}_2,\hv_3]$}
							\put(16,36.5){\small $\bH^{(2)}=[\hv_1,\hv_2,\check{\hv}_3]$}
							\put(20,0.5){\small $\bH^{(1)}=[\hv_1,\hv_2,\hv_3]$}
							\put(25.25,23.2){\small $\hv_1$}
							\put(26.7,26){\small $\hv_2$}
							\put(21,11.5){\small $\hv_1$}
							\put(13,3.5){\small $\hv_1$}
							\put(10.25,7.5){\small $\hv_3$}
							\put(7.8, 34){\small $\check{\hv}_2$}
							\put(8.75,30.7){\small $\check{\hv}_3$}
							\put(51,16.5){\small $\hv_1$}
							\put(54,27){\small $\hv_2$}
							\put(42,27){\small $\hv_3$}
							\put(33,32){\small $\bH^{(3)}=[\check{\hv}_1,\check{\hv}_2,\hv_3]$}							
							\put(52,31){\small $\bH^{(2)}=[\check{\hv}_1,\hv_2,\check{\hv}_3]$}							
							\put(44,0.5){\small $\bH^{(1)}=[\hv_1,\hv_2,\hv_3]$}
							\put(52,8){\small$\check{\bH}^{(1)}$}							
							\put(40.5,10.75){\small$\check{\bH}^{(1)}$}							
							\put(37.5,7.5){\small$\hv_3$}							
							\put(59,6.5){\small$\hv_2$}							
							\put(76,0.5){\small $\bH^{(1)}=[\hv_1,\hv_2,\hv_3]$}							
							\put(90,31){\small $\bH^{(2)}=\check{\bH}^{(1)}$}							
							\put(69.5,32){\small $\bH^{(3)}=\check{\bH}^{(1)}$}							
							\put(69.7,9){\small$\check{\bH}^{(1)}$}
							\put(90,7.5){\small$\check{\bH}^{(1)}$}
							\put(86,18){\small $\hv_1$}
							\put(83,22){\small $\hv_2$}
							\put(77,23.5){\small $\hv_3$}		
							\put(17,-2.5){\small $a)$}
							\put(49,-2.5){\small $b)$}
							\put(82,-2.5){\small $c)$}

						\end{overpic}~\vspace{2ex}
					\caption{Schematic illustration of three different example scenarios with distributed CSIT. FDD transmission is assumed. $\hv_i$ denotes the (highly accurate) estimate of user~$i$'s channel  that is fed back to the attached TX. ${\bH}^{(j)}$ denotes the CSI matrix obtained at TX~$j$ after cooperation. $\check{\bH}^{(j)}$ and $\check{\hv}_i$ represent a quantized/coarse version of the respective estimates, transmitted through limited backhaul communications.}\label{fig:examples}%
				\end{figure}

				\begin{example}\label{ex:ex1}
					In a network with several Base Stations (TXs) cooperating to jointly serve their users, each TX receives feedback from its attached user such that each TX knows accurately only a part of the channel state information. 
					For example, in Time Division Duplexing (TDD) transmission with reciprocity, each TX acquires a good estimate of its local channel; instead, in Frequency Division Duplexing (FDD) with user feedback, each TX acquires a good estimate of the whole channel vector towards its attached user. 
					
					Considering a CSI exchange step with heterogeneous links in the sense that one direction is of better quality, requires less quantization, or introduces less delay (which is a very reasonable assumption when considering heterogeneous networks where some of the TXs are UAVs~\cite{Becvar2017,Li2019} or vehicles~\cite{Veniam2019}), we obtain a setting where one of the TXs is uniformly more informed than the other. A particular example is depicted in~Fig.~\ref{fig:examples}.a, in which the link between TX~2 and TX~3 is a limited Device-to-Device (D2D) link. \qed
				\end{example}
				
				\begin{example}
					Let us consider the previous setting in the case where the CSI exchange is limited by a very restrictive delay. 
					Then, the sharing can be done by a transmission of the accurate CSI to a specific (main) TX, which then forwards a coarser version of the whole channel matrix to all TXs due to delay constraints. This retransmission from the main TX could also be broadcast, such that the resources spent on the sharing are reduced. Using layered encoding \cite{Ng2009}, every TX would obtain an estimate with a different accuracy. This setting is shown in~Fig.~\ref{fig:examples}.b.  \qed
				\end{example}					
				
				\begin{example}
					In a wireless network with one principal TX receiving feedback from all served users and several remote radio heads helping in the joint transmission, a distributed CSIT configuration is obtained when the CSI sharing from the main TX to the remote radio heads is done using limited and imperfect communication links, as illustrated in~Fig.~\ref{fig:examples}.c. 
					It could also be envisioned that the remote radio heads directly acquire low precision channel state information from direct feedback from the users using layered encoding \cite{Ng2009} or analog feedback \cite{ElAyach2012a}. Due to the weaker capabilities of the remote radio heads, a Distributed CSIT (D-CSIT) configuration with homogeneous quality at each TX would then be obtained.\qed
				\end{example}

				The analysis and the design of this CSI acquisition mechanism is an important research topic per-se and is consequently out of the scope of this paper and not further discussed. We assume in the following that this CSI acquisition step has already occurred through limited and imperfect communication links and has led to each TX having access to its own imperfect estimate of the multi-user channel state.

				\subsection*{User's Data Symbols Caching}
					Our second assumption comes from the fact that the user's data symbols are available at all TXs. This is made possible without putting in question the scenario described above because of two recent major techniques envisioned for future $5$G-and-beyond wireless networks: Caching \cite{Golrezaei2013,MaddahAli2014,Golrezaei2014} and Cloud-Ran/Fog-Ran \cite{Bonomi2012,Peng2015}. 
							
					Through caching, the user's data symbols are pre-fetched at the TX nodes before the transmission occurs \cite{MaddahAli2014}. Caching is an increasingly important feature that already exists in many scenarios \cite{Bastug2014,Zhang2017} and is envisioned in many more \cite{Chen2018}. With the user's data symbols available at the TXs, even at mobile and cost-efficient ones, the accurate and timely acquisition of the multi-user channel becomes the main bottleneck for efficient interference reduction. This leads to a D-CSIT configuration wherever the cooperation links are not of sufficient~quality.

					In the Cloud-RAN paradigm, the centralization of the processing of all nodes is envisioned so as to gain full benefits of cooperation. This centralization is however limited by its cost and its delay, such that partial centralization is considered a promising solution~\cite{iJOIND23}. 
					Considering decentralized precoding at the TXs allows to reduce the \emph{delay} in CSI acquisition. 
					In that case, the backhaul links are solely used to convey the user's data since, for many data-oriented applications, the application's latency requirements are orders of magnitude slower than the rate at which the fading channel evolves. The CSI estimates are directly exchanged between the TXs through direct links, thus reducing the delay of the complete CSI acquisition at the TXs. This CSI exchange between TXs through limited resources also leads to a D-CSIT configuration.
	
		\subsection*{Previous Works}
	 
			In \cite{Li2020}, the finite-SNR performance of regularized Zero-Forcing (ZF) under  distributed CSIT  has been computed in the large system limit, and heuristic robust precoding schemes have been provided in \cite{dekerret2014_ISWCS, Fritzsche2013b} for practical cellular networks. 
			In terms of DoF, it was shown in a previous work \cite{dekerret2012_TIT} that using a conventional ZF precoder (regularized or not) leads to a severe DoF degradation caused by the lack of a consistent CSI shared by the cooperating TXs. 
			Exploiting the fact that there is a single interfered user in the two-user case, it was shown that exploiting the CSIT at a single TX was sufficient to completely suppress the interference, thus recovering from the DoF loss due to the D-CSIT setting. 
			This result  was extended in \cite{Bazco2018_WCL} to the Generalized DoF (GDoF) \cite{Etkin2008}, where it was shown that the 2-user case centralized GDoF is recovered for any path-loss topology \emph{and  CSIT allocation}. 
			In \cite{dekerret2018_ISWCS}, a Deep Learning based robust precoder was proposed. 
			This approach has been then extended in \cite{Kim2018} to settings with a cooperation link between TXs and further generalized in \cite{Lee2019} to arbitrary topologies with local CSI at the TXs. An overview of the use of Machine Learning for the physical layer of wireless communications with a particular focus on such decentralized applications can be found in \cite{Gunduz2019}.
			
			The extension of the high SNR results for the two-user setting to an arbitrary number of users is not straightforward and has remained an open problem for several years. Tackling this challenging problem is the focus of this work.

			\subsection{Main Contributions}
			Our main contributions read as follows: 
				\begin{itemize}
						\item We compute the DoF of a genie-aided setting where all TXs are given the knowledge of all the channel estimates at all TXs, the so-called \emph{centralized upper bound}. 
						\item We show the surprising result that this bound is actually tight for a range of D-CSIT configurations, coined the \emph{Weak-CSIT regime} and defined rigorously further down. 
						Interestingly, the optimal DoF for such D-CSIT regime  only depends on the CSIT quality at the most informed TXs. 
						Sharing the instantaneous CSIT among the TXs is hence not necessary for achieving the genie-aided centralized DoF and does not improve the optimal DoF.
						\item  Building on the fundamental principles of the previous scheme, we derive a robust transmission scheme adapted to any CSIT configuration and any number of users, and which significantly improves the achieved DoF with respect to state-of-the-art methods. 
				\end{itemize}
				A byproduct of our work which completes our main contributions is the development of new methods used as building blocks to our main algorithm, and which are of interest by themselves for other applications. The first one is the non-trivial extension of the Active-Passive precoding scheme introduced in \cite{dekerret2012_TIT}. Specifically, we extend the scheme to multiple passive TXs which turns out to be essential to transmit multi-stream transmissions to a single-antenna user and hence create an overloaded transmission. The second method is the translation to the distributed CSIT setting of the idea introduced by Maddah-Ali and Tse in \cite{MaddahAli2012} consisting in estimating and retransmitting the interference generated. In contrast to \cite{MaddahAli2012}, the interference terms are estimated \emph{before} they even take place and are retransmitted in the same time slot. This principle could be applied in other wireless configurations where some nodes are more informed than others. 
						
		\paragraph*{Notations} We denote the probability density function of a variable $\Xc$ as $f_{\Xc}(x)$. 
				For a conditional probability density function, the simplified notation $f_{\Xc|\Yc}(x |y)$ stands for  $f_{\Xc|\Yc}(\Xc=x|\Yc=y)$.  
				The circularly symmetric complex Gaussian distribution with mean $\mu$ and variance $\sigma^2$ is denoted by~$\CN(\mu,\sigma^2)$. 
				We use $\doteq$ to denote exponential equality, i.e., we write $f(P) \doteq P^x$ to denote $\lim_{P\rightarrow \infty} \frac{\log f(P)}{\log P}= x$. 
				The exponential inequalities $\dotleq$ and $\dotgeq$ are defined in the same way. 
				We also consider the conventional Landau notation with a slight variation, such that $f(x)=O(g(x))$ stands for $\lim_{x\rightarrow\infty}  \frac{\abs{f(x)}}{\abs{g(x)}}=a$, with $0<a<\infty$, and $f(x)=o(g(x))$ stands for $\lim_{x\rightarrow\infty} \frac{\abs{f(x)}}{\abs{g(x)}}=0$.
				We use the shorthand notation $\mathcal{K}\triangleq\{1,\ldots,K\}$. 
				Given a matrix $\bA$, the sub-matrix obtained by taking from $\bA$ the rows $\{r_i,\dots,r_j\}$ and the columns $\{c_i,\dots,c_j\}$ is denoted by $\bA_{r_i:r_j, c_i:c_j}$.
				
	\section{System Model}\label{se:SM}

		\subsection{Transmission Model}
			We study a communication system where $K$~TXs jointly serve $K$~users (RXs) over a Network (Broadcast) MISO channel. 
			We assume a D-CSIT configuration, where each TX has access to its own CSIT, and we denote this setting henceforth as the $K$-user MISO BC with distributed CSIT.  
			We consider that each node (TX or RX) is equipped with a single-antenna. The assumption of single-antenna TX is done for ease of exposition, and the extension to multiple-antenna TX is straightforward.
			The signal received at RX~$i$ is written as
				\begin{align}
						y_i=\bm{h}_i^{\He}\xv+z_i,	\label{eq:SM_1}
				\end{align}
			where $\bm{h}_i^{\He}\in \mathbb{C}^{1\times K}$ is the channel vector towards RX~$i$, ${\xv\in\mathbb{C}^{K\times 1}}$ is the transmitted multi-user signal, which satisfies an average power constraint $\Exp[\norm{\xv}^2]\leq P$, and $P$ is the maximum transmit power. 
			$z_i\in \mathbb{C}$ is the additive noise at RX~$i$, independent of the channel and the transmitted signal, and distributed as~$\CN(0,1)$. 
			We further define the channel matrix~$\mathbf{H}\triangleq [\bm{h}_1,\ldots,\bm{h}_K]^{\He}\in \mathbb{C}^{K\times K}$ and the channel coefficient from TX~$k$ to RX~$i$ as $\bH_{i,k}$. 
			The channel is assumed to be drawn from a continuous distribution with density such that all the channel matrices and all their sub-matrices are full rank with probability one. 
			Furthermore, the channel coefficients are assumed to change after one channel use and to be independent from one channel use to another.   

			The transmitted multi-user signal~$\xv$ is obtained from the symbol vector $\bs \in  \mathbb{C}^{b\times 1}$ having its elements independently and identically distributed (i.i.d.) according to $\CN(0,1)$, where $b$ is the number of independent data symbols delivered. 
			We will differentiate in this work between the \emph{private} data symbols, destined to be decoded at a particular user, and the \emph{common} data symbols, broadcast and destined to be decoded at all users. Note that the term private is used only in contrast to common and does not refer to any privacy/secrecy constraint, but to the fact that only one user will decode the symbol.

			\subsection{Distributed CSIT Model}

			The Distributed CSIT (D-CSIT) setting differs from the conventional centralized one in that each TX receives a possibly different multi-user CSI, based on which it designs its own transmission parameters without any additional communication with the other TXs. 
			Specifically, TX~$j$ receives the imperfect multi-user channel estimate~$\hat{\bH}^{(j)}=[\hat{\bh}_1^{(j)},\ldots,\hat{\bh}^{(j)}_K]^{\He} \in \mathbb{C}^{K\times K}$ where $(\hat{\bh}^{(j)}_i)^{\He}$ refers to the estimate at TX~$j$ of the channel from all TXs towards RX~$i$. 
			TX~$j$ then designs its transmit coefficients solely as a function of~$\hat{\bH}^{(j)}$ and the statistics of the channel. This scenario can be seen as a multi-agent cooperative decision with common goal, where each node knows the structure of the system but not the information that the others own\cite{Radner1962}. 

				\begin{remark}
						It is critical to this work to understand well how the \emph{Distributed} CSIT setting differs from the many different \emph{heterogeneous} CSIT configurations studied in the literature. 
						Indeed, a heterogeneous CSIT configuration typically refers to a \emph{centralized} CSIT setting (i.e., with a channel estimate common to all TXs), where each element of the channel is known with a different quality owing to specific feedback mechanisms \cite{Lin2018_hybrid,Lashgari2016,Davoodi2018_ICC,Piovano2016}. In contrast, the distributed setting considered here has as many different channel estimates as TXs, and each TX does not have access to the CSIT knowledge at the other TXs.\qed 
				\end{remark}
			
				\begin{figure}[t]\centering~\\[3ex]
						~\hspace{25ex}\subfile{Bazco_TIT2018_FIG_setting_DCSIT_one}%
						\caption{$K\times K$ Distributed Broadcast Channel.  The accuracy of the channel estimate at TX~$j$ is modeled through the CSIT scaling coefficient~$\alpha^{(j)}$.} \label{DCSI_setting}
				\end{figure}
			
			We follow the conventional model in the literature~\cite{Gou2012,Yang2013,dekerret2016_ITW,Hao2017} to model the dependency of the CSIT accuracy as a function of the SNR,  and write the channel estimate at TX~$j$ as
				\begin{align}
						\hat{\bH}^{(j)}=\bH+{\Pb^{-\alpha^{(j)}}}\bm{\Delta}^{(j)},						\label{eq:SM_2a}
				\end{align}
			where $\Pb \triangleq \sqrt{P}$ and $P$ is the nominal SNR parameter. 
			${\bm{\Delta}^{(j)}\in\Cb^{K\times K}}$ is a random variable with zero mean and bounded covariance matrix.
			The scalar $\alpha^{(j)}$ is called the \emph{CSIT scaling coefficient} at TX~$j$ and represents the average accuracy of the estimate at TX~$j$. 
			Following insights obtained from the analysis of the centralized CSI configuration~\cite{Jindal2006,Davoodi2016}, we focus in this work on the configurations where~ $\alpha^{(j)}\in [0,1]$.
				\begin{remark} 
					It is shown in \cite{Davoodi2016} for the centralized case that if the  CSI accuracy does not improve polynomially with the SNR (i.e., a CSIT scaling coefficient equal to zero), the channel estimate leads to no DoF improvement. Similarly, a CSIT scaling coefficient equal to one is shown to be sufficient to attain the DoF obtained with perfect CSIT. These results are in fact very intuitive when considering the scaling of the interference. \qed
				\end{remark}

			For later use, we also denote the $i$-th row of $\bm{\Delta}^{(j)}$ as $(\bm{\delta}^{(j)}_i)^{\He}$, such that 
				\begin{align}
						\hat{\bh}^{(j)}_i=\bh_i+{\Pb^{-\alpha^{(j)}}}\bm{\delta}_i^{(j)}.				\label{eq:SM_2b}
				\end{align}
			We restrict in this work the D-CSIT model to an \emph{homogeneous CSIT quality} at a given TX, such that every channel coefficient is known at TX~$j$ with the same average CSIT quality~$\alpha^{(j)}$, as it can be seen in~\eqref{eq:SM_2a}. This limitation is not inherent to the {D-CSIT} model and it is solely done here for simplicity. How to deal with different CSIT qualities for the different channel coefficients is a topic of undergoing research and it is out of the scope of this work.

			Given that one TX has the same CSIT quality ($\alpha^{(j)}$) for all channel coefficients, we can order the TXs without loss of generality such that
				\begin{align}
						1 \geq \alpha^{(1)} \geq \alpha^{(2)}\geq \dots \geq \alpha^{(K)} \geq 0. \label{eq:distributed_alpha_order} 
				\end{align}
			The multi-user distributed CSIT configuration can be hence represented through the multi-TX CSIT scaling vector $\bm{\alpha}\in \mathbb{R}^{K}$ defined as
				\begin{align}
						\bm{\alpha}\triangleq \begin{bmatrix} \alpha^{(1)}, & \dots, & \alpha^{(K)}\end{bmatrix}^{\Transpose}.		\label{eq:SM_3}
				\end{align}    

			The parameters~$\bm{\alpha}$ represent the average accuracy of the estimates. They are long-term coefficients that vary slowly in time and can be easily shared to all TXs. Consequently, the parameters~$\bm{\alpha}$ are assumed in the following to be fixed and known at all TXs. 

			Our upper bound builds on the centralized upper bound of~\cite{Davoodi2016}. 
			On that account, and for the sake of completeness, we recall the Bounded Density assumption introduced in\cite{Davoodi2016}, which is necessary for the proof of the upper bound.  
				\begin{definition}[Bounded Density Coefficients\cite{Davoodi2016}]\label{def:bounded_density} 
						A set of random variables,~$\mathcal{A}$, is said to satisfy the bounded density assumption if there exists a finite positive constant $f_{\max}$,
						\begin{equation}
								0<f_{\max}<\infty,
						\end{equation}
						such that for all finite cardinality disjoint subsets~$\mathcal{A}_1$, $\mathcal{A}_2$ of $\mathcal{A}$, 
								$\mathcal{A}_1\subset \mathcal{A}$, $\mathcal{A}_2\subset \mathcal{A}$, $\mathcal{A}_1\cap \mathcal{A}_2=\emptyset$, $|\mathcal{A}_1|<\infty$, $|\mathcal{A}_2|<\infty$, 
						the conditional probability density functions exist and are bounded as follows
						\begin{equation}
								\forall A_1,A_2,~~f_{\mathcal{A}_1|\mathcal{A}_2}(A_1|A_2)\leq f_{\max}^{|\mathcal{A}_1|}.
						\end{equation}
				\end{definition}
			We assume hereinafter that all channel realizations~$\bH_{i,k}$ and estimation noise variables~$\bm{\Delta}_{i,k}^{(j)}$ satisfy the bounded density property. Furthermore, the channel realizations and the estimation noise are mutually independent.

		
			
			\subsection{CSIR Model}\label{se:SM:CSIR}
			In this work, we focus on the impact of the imperfect CSI on the TX side as the CSI acquisition is widely acknowledged to be more challenging on the TX side than on the RX side in FDD, because the CSI that has been estimated at the RX needs to be fed back towards the TX. Therefore, we consider that every RX has perfect knowledge of its own channel. As in the important literature on delayed CSIT \cite{MaddahAli2012,Tandon2012b,Gou2012,Chen2013a,Yang2013,dekerret2016_ITW}, to name just a few, we assume that the RX has been able to obtain perfect knowledge of the multi-user channel, i.e., also the channel to the other users. This assumption is key to the approach used in this work. 
						
			However, it is also important to note that this assumption can be weakened in our work as it is sufficient for the RXs to obtain the multi-user CSIT up to the best CSIT quality across the TXs (not necessarily the same estimate, but of the same quality). Furthermore, the estimate should be made available at the RX for the decoding, such that its latency constraint stems from the user's data, not from the time coherency of the channel.

			\FloatBarrier
			\subsection{Degrees-of-Freedom Analysis}
			We assume that every user~$i \in \mathcal{K}$ wishes to receive message $W_i\in\Wb_i$. After $n$~channel uses, the  rate $R_i(P)$ is achievable for RX~$i$ if $R_i(P)= \frac{\log|\Wb_i|}{n}$ and the probability of wrong decoding goes to zero as $n$ goes to infinity. The sum capacity~$\mathcal{C}(P)$ is  defined as the supremum of the sum of all achievable rates\cite{Cover2006}. The optimal sum DoF in the D-CSIT setting with CSIT scaling coefficients~$\bm{\alpha}\in\Rb^{K}$ is  defined by
				\begin{align}
						\DoF^{\DCSI}(\bm{\alpha})\triangleq \lim_{P\rightarrow \infty}\ \frac{\mathcal{C}(P)}{\log_2(P)}.				\label{eq:SM_7}
				\end{align} 

		\section{Main Results}\label{se:main}
		As a preliminary, let us first state the optimal DoF of the centralized $K$-user BC setting where a single  estimate with  CSIT scaling coefficient~$\alpha$ is perfectly shared by all TXs. It was shown in \cite{Davoodi2016} that the sum DoF in that configuration, denoted by $ \DoF^{\CCSI}(\alpha)$, is equal to 
			\begin{equation}
					\DoF^{\CCSI}(\alpha)=1+(K-1)\alpha.			\label{eq:centralized_DoF}
			\end{equation}
		We can now present our main results.
		\subsection{Upper bound}\label{subse:main_bound}
			\begin{theorem} \label{theo:centralized_upperbound}
					The optimal sum DoF of the $K$-user MISO BC with Distributed CSIT satisfies
						\begin{align} \label{eq:centralized_bound}
								\DoF^{\DCSI}(\bm{\alpha}) \leq \DoF^{\CCSI}(\alpha^{(1)}).
						\end{align}
			\end{theorem}
				\begin{proof}
					The proof relies on the following key lemma, whose proof is relegated to Appendix~\ref{app:bounded_estimation_proof} for clarity.
						\begin{lemma}\label{lem:density_bound}
								Let  $\hat{\Hc}^{(j)} \triangleq \Hc +  \Pb^{-\alpha^{(j)}}\Delta^{(j)}$, where $\Hc$,  $\Delta^{(j)},\ {\forall j\in\Kc}$,  are independent continuous random variables satisfying the Bounded Density assumption. Then, the conditional {probability~density} function $f_{\Hc|\hat{\Hc}^{(1)},\dots,\hat{\Hc}^{(K)}}$ satisfies that
									\begin{align}
											\max_{\Hc} f_{\Hc|\hat{\Hc}^{(1)},\dots,\hat{\Hc}^{(K)}} = O\LB{\Pb^{\max_{j\in \Kc}\alpha^{(j)}}}\RB. \label{eq:proof_bound_eq}
									\end{align}			
						\end{lemma}  
					Let  us now consider a genie-aided scenario where all channel estimates are exchanged between all the TXs. 
					Such setting corresponds to a (logically) centralized scenario with a shared CSI composed by $\{\hat{\bH}^{(1)},\dots,\hat{\bH}^{(K)}\}$. Using Lemma~\ref{lem:density_bound}, we obtain that the peak of the probability density function  of this genie-aided scenario with multiple estimates has the same scaling as the centralized setting endowed only with $\hat{\bH}^{(1)}$. 								
					It then directly follows from the proof in \cite[Section V.8]{Davoodi2016} that the DoF of the genie-aided scenario, denoted by $\DoF^{\CCSI}_{genie}\left(\bm{\alpha}\right)$, is given by
						\begin{align} \label{eq:genie_aided_bound}
								\DoF^{\CCSI}_{genie}\left(\bm{\alpha}\right) = \DoF^{\CCSI}\big(\alpha^{(1)}\big).
						\end{align}
					From this equivalence, and the fact that providing with more information does not hurt, the proof is concluded.
				\end{proof} 
			\begin{remark}
					Lemma~\ref{lem:density_bound} is expected to hold for a more general group of distributions, i.e., including cases where the different noise variables are partially correlated. Indeed, for the Gaussian case where the noise variables $\{\Delta^{(j)}_{i,k}\}_{\forall j\in\Kc}$ are drawn from partially correlated jointly Gaussian distributions, it is easy to show analytically that \eqref{eq:proof_bound_eq} is also satisfied.\qed
			\end{remark}

		
		\subsection{\hspace{-1.25ex}{1)} Achievability:   Weak-CSIT Regime}\label{subse:main_1}
				\begin{theorem}			\label{theo:weak_CSIT}
						Let us assume that the $m$ first TXs have the same accuracy, i.e., that $\alpha^{(1)}=\dots=\alpha^{(m)}$,  $m<K$. 
						Then, the sum DoF of the K-user MISO BC with Distributed CSIT satisfies
							\begin{equation}
									\DoF^{\DCSI}(\bm{\alpha}) \geq \DoF^{\CCSI}(\alpha^{(1)}) 					\label{eq:distributed_opt_dof_K_2}
							\end{equation}								
						if  $\alpha^{(1)}\leq \alpha^{\Weak}_m$, where $\alpha^{\Weak}_m$ is defined as
							\begin{align}
									\alpha^{\Weak}_m\triangleq \frac{1}{1+K(K-m-1)}, 
							\end{align}
						and is called the ``$m$-TX Weak-CSIT'' regime. For $m=1$, we simplify the notation and call it the ``Weak-CSIT'' regime.
				\end{theorem} 						
				\begin{proof}
						The result follows directly from the analysis of the proposed scheme presented in detail in  Section~\ref{se:APZF}. 
				\end{proof}		
			In the so-called Weak-CSIT regime, the upper bound of Theorem~\ref{theo:centralized_upperbound} is tight.	Surprisingly, for $m=1$, the most heterogeneous case, the DoF depends only on the CSI quality of TX~$1$, although with the downside of reducing the range of possible CSIT configurations. 
			For $m=K-1$, it holds that $\alpha^{\Weak}_{K-1}=1$ and hence the Weak-CSIT regime encloses every possible value of $\alpha\expo$, which is consistent with the simple use of single-stream Active-Passive Zero-Forcing (AP-ZF) in \cite{dekerret2012_TIT}. 
				\begin{remark}
						The fact that it is possible to achieve the DoF of the centralized upper bound with uniformed or badly informed TXs is a surprising result which is not expected to extend to many other CSIT configurations. Indeed, it can be intuitively seen using basic linear algebra that at least $K-1$ TXs are necessary to cancel $K-1$ ZF constraints. \qed
				\end{remark}

		\setcounter{subsection}{1}
			\subsection{\hspace{-1.25ex}{2)} {} Achievability: Extension to Arbitrary CSIT Configurations}\label{subse:main_2}
			Theorem~\ref{theo:weak_CSIT} presents the CSIT regime for which the upper bound of Theorem~\ref{theo:centralized_upperbound} is tight. In the following, we present a robust transmission scheme that builds on the scheme attaining Theorem~\ref{theo:weak_CSIT} but which is extended to adapt to any CSIT configuration. 
			The main challenge comes from the large number of CSIT scaling parameters, which leads to an even larger (combinatorially large) number of possible CSIT configurations depending on their relative values. Before presenting the achievability result, we define three terms that play an important role in the proposed scheme.

				\begin{definition}[Transmitting TXs]
						A TX is said to be a ``Transmitting TX'' if it is connected and sends information to the RXs. It may or may not use its instantaneous CSI for precoding. 
				\end{definition}				
										
			This definition is made necessary by the distributed nature of the CSIT. Indeed, in contrast to the centralized setting where adding antennas cannot reduce the performance \cite{Davoodi2018_ICC,Vaze2012b}, using additional antennas in the distributed setting can decrease the achievable DoF by creating additional interference. Hence, although we are considering a setting with $K$ TXs, it may be beneficial in some CSI configurations to ``turn off'' some  TXs and  use a smaller number of ``Transmitting TXs''.
				
				\begin{definition}[Active TXs]
						A TX is said to be an ``Active TX'' if it is a Transmitting TX and it makes use of its instantaneous CSI. 
				\end{definition}
							
				\begin{definition}[Passive TXs]
						A TX is said to be a ``Passive TX'' if it is a Transmitting TX but it does not make use of its instantaneous CSI. 
				\end{definition} 
			
			A more thorough explanation of these definitions and their relevance is provided later on, along with the description of the proposed scheme. 
			In brief, it will become clear that the two critical parameters to optimize are both the number of  ``Transmitting TXs''  and the number of ``Active TXs''. In relation to these two notions, we introduce the following definition.

				\begin{definition}[Transmission Mode~$(n,k)$]
					We define the Transmission Mode~$(n,k)$ as a transmission where only $k$~Transmitting TXs and $n\leq k$ Active TXs are used. 
				\end{definition}
				
			Building on these definitions, the following lower bound is exactly obtained by optimizing the performance of the proposed scheme over the different Transmission Modes.
						
			\begin{theorem}	\label{theo:theorem_achievable}
					The sum DoF of the $K$-user D-CSIT BC with CSIT scaling coefficients~$\bm{\alpha}$ is lower-bounded by $\DoF^{\APZF}(\bm{\alpha})$, obtained by solving the next linear program: 
						\begin{align}
								&\DoF^{\APZF}(\bm{\alpha})\ =\ \underset{\substack{\gamma_{n,k}}}{\mathrm{maximize}}~~ \sum_{k=2}^{K}~\sum_{n=1}^{k-1}\gamma_{n,k} \Big( 1+(k-1)\alpha^{(n)}\Big) \label{eq:Objective_primo}\\
									&\hspace{15ex}~~\mathrm{subject~to}~~ \sum_{k=2}^{K}~\sum_{n=1}^{k-1}\gamma_{n,k}=1, \hspace{2ex}\gamma_{n,k}\geq 0 \label{eq:time_sharing}\\
									&\hspace{17ex}\phantom{subject~to}~\sum_{k=2}^{K}~\sum_{n=1}^{k-1}d_{n,k}\gamma_{n,k}   \geq 0, \label{eq:interference_retransmission_condition} 
						\end{align}  
					where $\gamma_{n,k}$ is a time-sharing variable representing the percentage of time allocated to the Transmission Mode~$(n,k)$, and $d_{n,k} \triangleq  1 - \alpha^{(n)} - k(k-n-1)\alpha^{(n)}$.
			\end{theorem}
				\begin{proof}
					The transmission scheme for a particular Transmission Mode is described in Section~\ref{se:APZF} and, building on this result, the explanation and proof of the theorem is given in Section~\ref{se:Full_scheme}.
				\end{proof}
			The transmission scheme and the achieved DoF are obtained by solving a simple linear programming problem with low complexity. 			
			Interestingly, the expression~$ 1+(k-1)\alpha^{(n)}$ in \eqref{eq:Objective_primo} corresponds to the DoF achieved in the $k$-user centralized setting with $k$~TXs having received a CSIT of quality~$\alpha^{(n)}$ (See \eqref{eq:centralized_DoF}).
				\begin{remark} \label{corollary_homogeneous_CSIT}
					The linear program of Theorem~\ref{theo:theorem_achievable} depends only on the $K-1$~best CSIT coefficients and not on $\alpha^{(K)}$. 
					This property was already highlighted in \cite{dekerret2012_TIT} and follows from the fact that it is possible to solve $K-1$ linear equations with $K-1$~Active TXs, which means removing interference at $K-1$~users, and thus serving $K$~users at the same time. As consequence, it can always be assumed that $\alpha^{(K)}=0$ without reducing the DoF.\qed
				\end{remark}	
				Furthermore, we can show that time sharing between only two Transmission Modes is sufficient.
				\begin{corollary}\label{theorem_two_phases}
					The optimal solution of Theorem~\ref{theo:theorem_achievable} is composed only of two Transmission Modes,~$(n_1,k_1)$ and $ (n_2,k_2)$, such that 
						\eqm{
							\gamma_{n_1,k_1}> 0,\ \gamma_{n_2,k_2}\geq 0,\ \gamma_{n,k}= 0,
						}
					for any pair $(n,k)\notin \{(n_1,k_1), (n_2,k_2)\}$.  
				\end{corollary}
				\begin{proof}
						The proof is relegated to Appendix~\ref{app:two_phases}.
				\end{proof}
				\begin{remark}				
					The Distributed CSIT setting here assumed and the Delayed CSIT setting share an important feature: It is not possible to cancel out all the interference. 
					However, the means for solving this challenge at each setting differ considerably. 
					For the Distributed CSIT setting, the best strategy consists on overloading the transmission so as to use the interference created as side information that is useful at different RXs. 
					In addition, the more RXs receive interference, the more overloaded the transmission has to be. Corollary~\ref{theorem_two_phases} reflects the main insight: 
					If there are two modes of transmission, the first one is a generator of interference, i.e., it creates side-information at the RXs through the overloaded transmission.
					Then, it relies on a successive second Transmission Mode to retransmit some of this interference (side information) in order to decode the overloaded transmission. 
					When only one mode is used, the interference is directly retransmitted through rate splitting using the common data symbol.\qed
				\end{remark}		

				\begin{figure}[t]\centering
					\begin{tikzpicture}
							\begin{axis}[	width=0.65*0.95*0.94\columnwidth,	height=0.65*0.86*0.87\columnwidth,	at={(0\columnwidth,0\columnwidth)},
										scale only axis,	separate axis lines,
										every outer x axis line/.append style={black},
										every x tick label/.append style={font=\color{black}},
										xmin=0,	xmax=1,
										every outer y axis line/.append style={black},
										every y tick label/.append style={font=\color{black}},
										ymin=.0,	ymax=4,
										samples=11,
										ylabel={Sum DoF},	xlabel=$\alpha^{(1)}$,	
										xtick={ 0, 0.1,0.2,0.3,0.4,0.5,0.6,0.7,0.8,0.9,1},											
										axis background/.style={fill=white},
										y label style={at={(axis description cs:-0.08,.51)},rotate=0,anchor=south},
										x label style={at={(axis description cs:0.5,-.07)}},
										axis background/.style={fill=white}, grid=both, grid style={line width=.1pt, draw=gray!10},	major grid 	style={line width=.2pt,draw=gray!50}, 
										legend style={at={(0.025,0.975)},draw=none,anchor=north west}
							]

								\addplot[dashed, line width = 2.5pt,domain=0:1] {1+3*\x} ;

								\addplot[solid, line width = .75pt,domain=0.2:1,mark=square*,mark options={solid,scale=1.5,fill=cyan,draw=black},draw=black,domain=0.2:1,samples=9] plot (\x,{1.0075 +	2.9625*\x}) node[right] {};					
								\addplot[solid, line width = .75pt,domain=0:1,mark=triangle*,mark options={solid,scale=2.5,fill=red,draw=black},draw=black,domain=0.2:1,samples=9]{1.1775+2.1125*\x};
								\addplot[thick,dash pattern={on 7pt off 2pt on 1pt off 3pt}, line width = 1.75pt,domain=0.2:1,draw=black,domain=0.2:1,samples=9] plot (\x,{1.25 +	1.75*\x}) node[right] {};				
								
								\addplot[solid, line width = .75pt,domain=0:1,mark=x,mark options={solid,scale=2,draw=black}] {1} ;
										
								\legend{ Centralized Upper bound, $\alpha^{(3)}=0.99\alpha^{(1)}$, $\alpha^{(3)}=0.75\alpha^{(1)}$ ,Lower bound $\alpha^{(3)} = 0$\vspace{-.5ex}}
									
								\addplot[solid, line width = .75pt,domain=0:1,mark=square*,mark options={solid,scale=1.5,fill=green!75!black,draw=black},draw=black,domain=0:0.2,samples=3]{1+3*\x};
								\addplot[solid, line width = .75pt,domain=0:1,mark=triangle*,mark options={solid,scale=2.5,fill=green!75!black,draw=black},draw=black,domain=0:0.2,samples=3]{1+3*\x};
								\addplot[solid, line width = .75pt,domain=0:1,mark=*,mark options={solid,scale=1.5,fill=green!75!black,draw=black},draw=black,domain=0:0.2,samples=3]{1+3*\x};
								
								\draw[draw=black!90!white,  thick, densely dotted] (0.2,0.0) -- (0.2,4);
								\node at (.1, 1.93) [rectangle, align=center, text=blue, draw=white!50!white] (v1p00) {\small Weak-CSIT\\ \small Region};
								\draw[<->] (.005,2.21) -- (0.195,2.21);

								\node at (.45, 1.2) [rectangle, align=center, text=black, draw=white!50!white] (v100) {Zero-Forcing $\forall \alpha^{(3)}$};
								
								\draw[black, fill = cyan, fill opacity = 0.15, semithick, even odd rule]
								(0,0) rectangle (0.2,4) ;
							\end{axis}
					\end{tikzpicture}%
					\caption{Sum DoF for the case with $K=4$~TXs, $\alpha^{(2)}=\alpha^{(1)}$, and $\alpha^{(4)} =0$, for different values of $\alpha^{(3)}$ as a function of $\alpha^{(1)}$.}\label{fig:DoF_evolution}\vspace{-2ex}
				\end{figure}

			We show in Fig.~\ref{fig:DoF_evolution} the DoF as a function of~$\alpha^{(1)}$ for a setting with $K=4$~TXs and $\alpha^{(2)}=\alpha^{(1)}$, and we compare the achievable DoF with the centralized upper bound for different values of~$\alpha^{(3)}$. 
			The centralized upper bound is achieved for any $\alpha\expo\leq 0.2 = \alpha^{\text{Weak}}_2$, no matter the value of $\alpha^{(3)}$ (as stated in  Theorem~\ref{theo:weak_CSIT}). 
			The upper bound is also achieved when $\alpha^{(3)}$ becomes equal to $\alpha^{(1)}$, which is  consistent with the results in \cite{dekerret2012_TIT}. 

			In Fig.~\ref{fig:dof_casesTX}, we show the DoF achieved by AP-ZF for a network with $K=4$~TXs when fixing the number of Transmitting TXs (i.e., the value of $k$ in Theorem~\ref{theo:theorem_achievable}) for the specific case where $\alpha^{(1)}=1$, $\alpha^{(3)}=\alpha^{(4)}=0$, and $\alpha^{(2)}$ varies from $0$ to $1$. Depending on the value of $\alpha^{(2)}$, it is optimal to use either $2$ Transmitting TXs or $3$ Transmitting TXs, but never $K=4$ Transmitting TXs.
			
				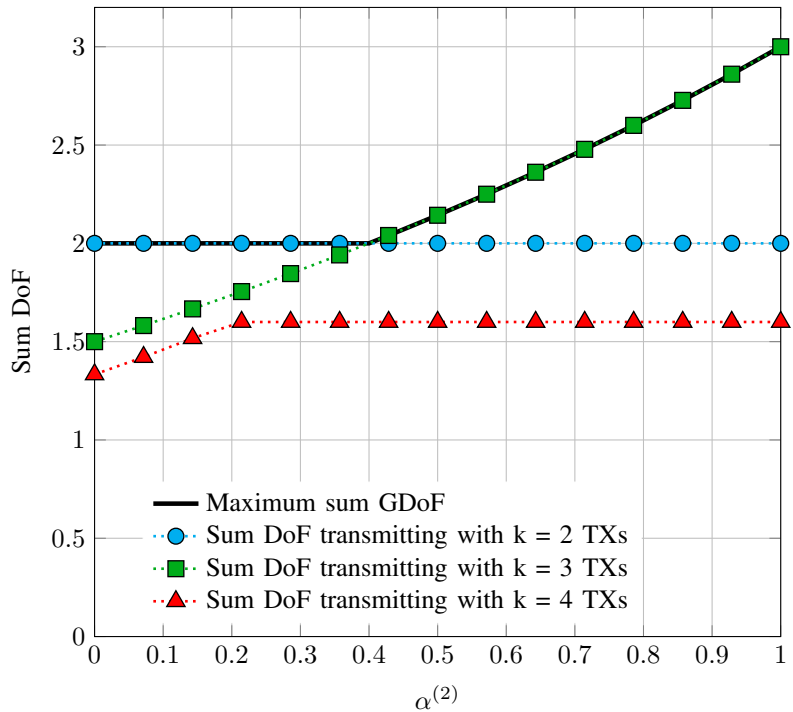
\begin{figure}[tb] \centering		
					\begin{tikzpicture}
								\begin{axis}[%
										width=0.65*0.85\columnwidth,
										height=0.65*.78\columnwidth,
										at={(-0,0)},
										scale only axis,
										xmin=0,
										xmax=1,
										x label style={at={(axis description cs:0.5,-.13)},anchor=south},
										y label style={at={(axis description cs:-0.08,.5)},anchor=south},
										xlabel={$\alpha^{(2)}$},
										xmajorgrids,
										ymin=0,
										ymax=3,
										ylabel={Sum DoF},
										ymajorgrids,
										axis background/.style={fill=white},
										legend style={at={(0.07,0.021)},anchor=south west,legend cell align=left,align=left,draw=white!15!white}
								]
											\addplot [color=black,line width = 1.725pt]
												table[row sep=crcr]{%
											0	2\\
											0.0714285714285714	2\\
											0.142857142857143	2\\
											0.214285714285714	2\\
											0.285714285714286	2\\
											0.357142857142857	2\\
											0.400000000000000	2\\
											0.428571428571429	2.03999999978014\\
											0.5	2.14285714265853\\
											0.571428571428571	2.24999999982952\\
											0.642857142857143	2.36170212752294\\
											0.714285714285714	2.47826086946582\\
											0.785714285714286	2.59999999993762\\
											0.857142857142857	2.72727272724246\\
											0.928571428571429	2.86046511627554\\
											1	3\\
											};
											\addlegendentry{Maximum sum GDoF};

											\addplot [color=cyan, dotted, mark size=3pt, line width = 1pt, mark=*,mark options={solid,fill=cyan,draw=black,line width=.5pt}]
												table[row sep=crcr]{%
											0	2\\
											0.0714285714285714	2\\
											0.142857142857143	2\\
											0.214285714285714	2\\
											0.285714285714286	2\\
											0.357142857142857	2\\
											0.428571428571429	2\\
											0.5	2\\
											0.571428571428571	2\\
											0.642857142857143	2\\
											0.714285714285714	2\\
											0.785714285714286	2\\
											0.857142857142857	2\\
											0.928571428571429	2\\
											1	2\\
											};
											\addlegendentry{Sum DoF transmitting with k = 2 TXs};

											\addplot [color=green!75!black!90!blue,dotted, mark size=3pt, line width = 1pt, mark=square*,mark options={solid,fill=green!75!black!90!blue,draw=black,line width=.5pt}]
												table[row sep=crcr]{%
											0	1.49999999981254\\
											0.0714285714285714	1.58181818160033\\
											0.142857142857143	1.66666666642696\\
											0.214285714285714	1.75471698088902\\
											0.285714285714286	1.84615384591237\\
											0.357142857142857	1.94117647035423\\
											0.428571428571429	2.03999999978014\\
											0.5	2.14285714265853\\
											0.571428571428571	2.24999999982952\\
											0.642857142857143	2.36170212752294\\
											0.714285714285714	2.47826086946582\\
											0.785714285714286	2.59999999993762\\
											0.857142857142857	2.72727272724246\\
											0.928571428571429	2.86046511627554\\
											1	3\\
											};
											\addlegendentry{Sum DoF transmitting with k = 3 TXs};

											\addplot [color=red,dotted, mark size=4pt, line width = 1pt, mark=triangle*,mark options={solid,fill=red,draw=black,line width=.5pt}]
												table[row sep=crcr]{%
											0	1.33333333332171\\
											0.0714285714285714	1.42148760324936\\
											0.142857142857143	1.51724137930266\\
											0.214285714285714	1.59999999641307\\
											0.285714285714286	1.6\\
											0.357142857142857	1.59999999999323\\
											0.428571428571429	1.59999999610031\\
											0.5	1.59999999998144\\
											0.571428571428571	1.5999999999134\\
											0.642857142857143	1.59999999091723\\
											0.714285714285714	1.6\\
											0.785714285714286	1.59999999999999\\
											0.857142857142857	1.59999999999999\\
											0.928571428571429	1.59999999999999\\
											1	1.59999999999999\\
											};
											\addlegendentry{Sum DoF transmitting with k = 4 TXs};
						\end{axis}
					\end{tikzpicture}%
					\caption{Sum DoF of the $K=4$ MISO BC with distributed CSIT as a function of $\alpha^{(2)}$, with $\alpha^{(1)}=1$, and $\alpha^{(3)}=\alpha^{(4)}=0$.} \label{fig:dof_casesTX}
				\end{figure}%
					
		\section{A Simple Example}\label{se:example}
		We present in the following a simple transmission scheme in a toy example so as to exemplify the key features of our approach and convey the main intuition in a clear manner. The full scheme achieving the DoF of Section~\ref{se:main} will be described in Section~\ref{se:APZF}. 

		We consider a $3$-user setting with $\alpha^{(1)}=0.1$, $\alpha^{(2)}=0$, and $\alpha^{(3)}=0$. 
		The conventional regularized Zero-Forcing would achieve a DoF of $1$\cite{dekerret2016_ISIT}, which is the same as for the no CSIT scenario. 
		We will show how it is possible to achieve a DoF of $1+2\alpha^{(1)}=1.2$, which is the DoF that would be achieved in a centralized setting if TX~$2$ and TX~$3$ had received the same estimate as TX~$1$ \cite{Davoodi2016}, such that there is no DoF loss from \emph{not} sharing the CSIT between the TXs.
		
					\begin{figure*}[t]
						\centering
							\begin{overpic}[width=0.99\textwidth]{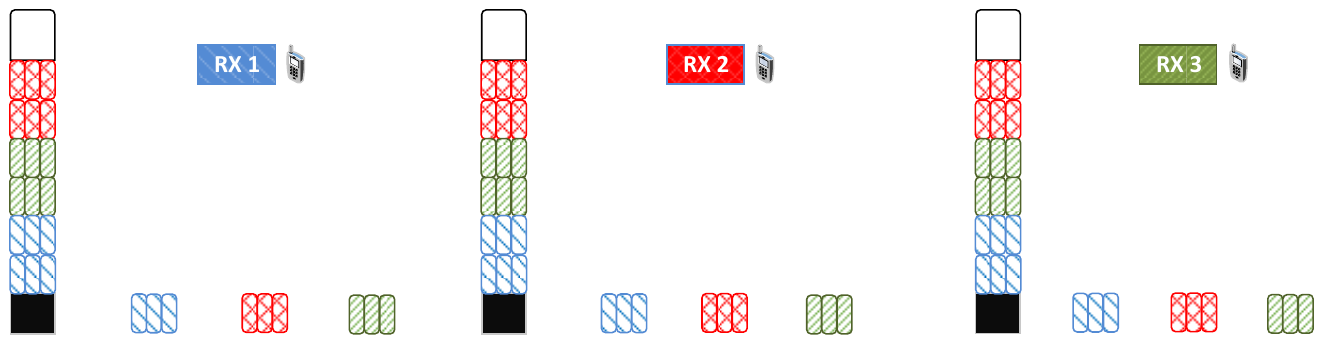}					
									{\put(0,-2){\parbox{\linewidth}{\resizebox{!}{0.95\height}{$\ \bH_{1,1}s_0\, +\, \bh^{\He}_1\bs_1 \, +\, \bh^{\He}_1\bs_2 \, +\  \bh^{\He}_1\bs_3$}}}}
									{\put(35,-2){\parbox{\linewidth}{\resizebox{!}{0.95\height}{$\bH_{2,1}s_0\ +\ \bh^{\He}_2\bs_1 \ + \bh^{\He}_2\bs_2 \ +\,   \bh^{\He}_2\bs_3$}}}}
									{\put(72.5,-2){\parbox{\linewidth}{\resizebox{!}{0.95\height}{$\,\bH_{3,1}s_0  +  \bh^{\He}_3\bs_1  + \bh^{\He}_3\bs_2  +  \bh^{\He}_3\bs_3$}}}}
									
									{\put(19,4.5){\parbox{\linewidth}{$\scriptstyle b_1$}}}
									{\put(27.5,4.5){\parbox{\linewidth}{$\scriptstyle c_1$}}}		
									{\put(46.5,4.5){\parbox{\linewidth}{$\scriptstyle a_1$}}}
									{\put(62,4.5){\parbox{\linewidth}{$\scriptstyle c_2$}}}
									{\put(82,4.5){\parbox{\linewidth}{$\scriptstyle a_2$}}}
									{\put(89.5,4.5){\parbox{\linewidth}{$\scriptstyle b_2$}}}
														
									{\put(4.5,4.5){\parbox{\linewidth}{$\scriptstyle a_1$}}}
									{\put(4.5,7.5){\parbox{\linewidth}{$\scriptstyle a_2$}}}
									{\put(4.5,10.75){\parbox{\linewidth}{$\scriptstyle c_1$}}}
									{\put(4.5,13.75){\parbox{\linewidth}{$\scriptstyle c_2$}}}
									{\put(4.5,16.5){\parbox{\linewidth}{$\scriptstyle b_1$}}}
									{\put(4.5,19.5){\parbox{\linewidth}{$\scriptstyle b_2$}}}
														
									{\put(49.5,24.5){\parbox{\linewidth}{\resizebox{!}{.8\height}{$\Pb$\ }}}}				
									{\put(40,24.75){\parbox{\linewidth}{\resizebox{!}{.8\height}{-----------------}}}}				
									{\put(64.5,3.65){\parbox{\linewidth}{\resizebox{!}{.8\height}{---$\Pb^{\alpha^{(1)}}$\!\!\!\!---}}}}				
									{\put(64.5,1.65){\parbox{\linewidth}{\resizebox{!}{1.6\height}{$\uparrow$}}}}				
									{\put(69,1.65){\parbox{\linewidth}{\resizebox{!}{1.6\height}{$\uparrow$}}}}				
									{\put(33,12.55){\parbox{\linewidth}{$\scriptstyle s_0\begin{cases}~\\~\vspace{7ex}\\~\\~\\\end{cases}$}}}

									\linethickness{0.5mm} \color{blue!85!white}%
											\put(0,0){\polygon(-1.5,-6)(-1.5,26)(30,26)(30,-6)}
									\linethickness{0.5mm} \color{red!85!black}%
											\put(0,0){\polygon(30.5,-6)(30.5,26)(69,26)(69,-6)}
									\linethickness{0.5mm} \color{green!75!black!90!blue}%
											\put(0,0){\polygon(69.5,-6)(69.5,26)(97.5,26)(97.5,-6)}																		
							\end{overpic}	
							~\\ \vspace{7ex} 
							\caption{Illustration of the received signals. Each RX receives its desired private data symbols and interference scaling both in $\Pb^{\alpha^{(1)}}$. Through superposition coding, it also receives the common data symbol~$s_0$ containing a mix of fresh desired data symbols (illustrated in white), and side information to remove interference (illustrated with the color of the relevant RX).} \label{fig:ISIT2016_first_scheme}
					\end{figure*}
						
			\subsection{Encoding}

			The transmission scheme consists in a single channel use during which $3$~private data symbols of  $\alpha^{(1)}\log_2(P)$~bits are sent to each user, thus leading to $9$~data symbols being sent in one channel use.
			Additionally, a common data symbol of rate $(1-\alpha^{(1)})\log_2(P)$~bits is broadcast from TX~$1$ to all users using superposition coding~\cite{Cover2006}. Note that the information contained in this common data symbol is not only composed of ``fresh'' information bits destined to one user, but is also composed of side information necessary for the decoding of the private data symbols, as will be detailed below.  
      The transmitted signal~$\xv\in \mathbb{C}^3$ is then equal to
				\begin{equation}
						\xv=\bs_1+\bs_2+\bs_3+\begin{bmatrix}1\\0\\0\end{bmatrix} s_0
				\end{equation}
			where
				\begin{itemize}
						\item $\bm{s}_i\in \mathbb{C}^3$ is a vector containing the 3  private data symbols destined to user~$i$, each one with power $P^{\alpha^{(1)}}/9$ and rate $\alpha^{(1)}\log_2(P)$~bits.
						\item $s_0$ is the common data symbol transmitted \emph{only from TX~$1$} and destined to be decoded at all users. 
						It is transmitted with power $P-P^{\alpha^{(1)}}$ and rate $(1-\alpha^{(1)})\log_2(P)$~bits. 
				\end{itemize}
			The signal received at RX~$i$, illustrated in Fig.~\ref{fig:ISIT2016_first_scheme}, is  given by
				\begin{equation} 
						\!y_i\!=\!\underbrace{\bH_{i,1}s_0}_{\doteq P}+\underbrace{\bh_i^{\He}\bm{s}_{1}}_{\doteq P^{\alpha^{(1)}}} +\underbrace{\bh_i^{\He}\bm{s}_{2}}_{\doteq P^{\alpha^{(1)}}} +\underbrace{\bh_i^{\He}\bm{s}_{3}}_{\doteq P^{\alpha^{(1)}}} \label{eq:example_1}
				\end{equation}
			where the power scaling is written under the bracket, and where the noise term has been neglected for clarity.

			\subsection{Interference Estimation and Quantization at TX~$1$} 
			The key element of the scheme is that the common data symbol~$s_0$ is used to convey side information, enabling each user to decode its desired private data symbols. More specifically, TX~$1$ uses its local CSIT~$\hat{\bH}^{(1)}$ to \emph{estimate} the interference terms $(\hat{\bh}^{(1)}_i)^{\He}\bm{s}_k$ generated by the first layer of transmission, for any $i, k$ such that $k\neq i$.  
			Then, TX~1 quantizes  and transmits these interference terms using the common data symbol~$s_0$. Each interference term has a variance scaling in~$\Pb^{\alpha^{\scalebox{.75}{\tiny$(1)$}}}$
			and is quantized using $\alpha^{(1)}\log_2(P)$~bits, such that the quantization noise can be made to remain at the noise floor using an appropriate uniform or  Lloyd quantizer~\cite{Cover2006}. 
			In total, the transmission of all the quantized estimated interference requires a transmission of~$6\alpha^{(1)}\log_2(P)$~bits.

			These  $6\alpha^{(1)}\log_2(P)$~bits can be transmitted via the data symbol~$s_0$ if $6\alpha^{(1)}\log_2(P) \leq (1-\alpha^{(1)})\log_2(P)$, which is the case for the example considered here since $6\times 0.1<1-0.1$. If the inequality is strict (as it is in this case),  $s_0$ carries some additional $(1-7\alpha^{(1)})\log_2(P)$ fresh-information bits  to any particular user. 
			 
			\subsection{Decoding and DoF Analysis}
			It remains to verify that this transmission scheme achieves the claimed DoF. 
			Let us consider w.l.o.g. the decoding at RX~$1$, because the decoding at the other users is the same up to a circular permutation of the RX's indices. Note that signals at the noise floor are systematically omitted.
			
			Using successive decoding~\cite{Cover2006}, the common data symbol~$s_0$ is decoded first, followed by the private data symbols~$\bs_1$. The data symbol~$s_0$ of rate of $(1-\alpha^{(1)})\log_2(P)$~bits can be decoded with a vanishing probability of error as its effective SNR can be seen in \eqref{eq:example_1} to scale in $P^{1-\alpha^{(1)}}$. 
			Upon decoding~$s_0$, the quantized estimated interference $(\hat{\bh}^{(1)}_1)^{\He}\bm{s}_2$ is obtained up to the quantization noise. 
			RX~$1$ has then decoded
				\begin{align}
						&\underbrace{(\hat{\bh}^{(1)}_1)^{\He}\bm{s}_2}_{\doteq P^{\alpha^{(1)}}} = \bh_1^{\He}\bm{s}_2+\underbrace{{\Pb^{-\alpha^{(1)}}}\!(\bm{\delta}^{(1)}_1)^{\He}\bm{s}_2}_{\doteq P^{0}}.				\label{eq:example_2}
				\end{align}
			Since the quantization noise is at the noise floor, it is neglected in the following. 
			This means that the interference term $\bh_1^{\He}\bm{s}_2$ can be suppressed up to the noise floor at RX~$1$.	Proceeding in the same way with $(\hat{\bh}^{(1)}_1)^{\He}\bm{s}_3$, the remaining signal at RX~$1$ is 				\begin{equation}
						y_1= \bh_1^{\He}\bm{s}_1.				\label{eq:example_3}
				\end{equation}
			RX~$1$ can form a virtual received vector~$\by_1^{\mathrm{v}}\in\mathbb{C}^{3}$ by bringing together the signal in~\eqref{eq:example_3} and the  interference terms $(\hat{\bh}_2^{(1)})^{\He}\bm{s}_1$  and $(\hat{\bh}_3^{(1)})^{\He}\bm{s}_1$ obtained through~$s_0$. Therefore,~$\by_1^{\mathrm{v}}\in\mathbb{C}^{3}$ is defined as
				\begin{equation}
						\by_1^{\mathrm{v}}\triangleq 
							\begin{bmatrix}
									\bh_1^{\He}\\
									(\hat{\bh}_2^{(1)})^{\He} \\
									(\hat{\bh}_3^{(1)})^{\He}
							\end{bmatrix}\bm{s}_1.			\label{eq:example_4}
				\end{equation}
			Each component of $\by_1^{\mathrm{v}}$ has an effective SNR scaling in $P^{\alpha^{(1)}}$ such that RX~$1$ can decode with a vanishing error probability its $3$~data symbols of rate $\alpha^{(1)}\log_2(P)$~bits.

			Considering the $3$~users, $9\alpha^{(1)}\log_2(P)$~bits have been transmitted through the private data symbols and $(1-7\alpha^{(1)})\log_2(P)$~bits through the common data symbol~$s_0$, which yields a sum DoF of $1+2\alpha^{(1)}$. 
				\begin{remark}
						Interestingly, the above scheme is based on interference estimation, quantization, and retransmission, in a similar fashion as done in the different context of precoding with delayed CSIT (see e.g. \cite{Gou2012,Yang2013,dekerret2016_ITW}). Yet, we exploit in this work the distributed nature of the CSIT instead of the delayed knowledge of the CSIT, such that the scheme  estimates and quantizes the interference \emph{before even being generated}.\qed
				\end{remark}

	\FloatBarrier
		\section{Transmission Mode~$(n,k)$ with $n$~Actives TXs and $k$ Transmitting TXs}\label{se:APZF}
		We present in this section the Transmission Mode $(n,k)$ with $n$~Active TXs and $k$ Transmitting TXs. 
		We  describe the main structure of the transmission in Section~\ref{se:weak_CSIT:encoding}, before describing in detail the precoding scheme in Section~\ref{subse:apzf_expl}. 
		The received signal is then studied in Section~\ref{se:RX_signal}, and the achieved DoF is computed in Section~\ref{se:weak_CSIT:quantization} and  Section~\ref{se:DoF}.
		As main insight, a transmission using $n$~\emph{Active TXs} can reduce the interference power received at $n$~RXs by a factor $P^{\alpha^{(n)}}$, i.e., as $n$ increases, it holds that:
		\begin{itemize}
			\item The interference power is reduced by a \emph{smaller} factor.
			\item The interference power is reduced at \emph{more} RXs.
		\end{itemize} 

			\subsection{Encoding}\label{se:weak_CSIT:encoding}
			The Transmission Mode $(n,k)$ only considers $k$~Transmitting TXs, and thus only $k$~RXs are served simultaneously.
			Let $\Uc$ denote the set of RXs served and let us assume w.l.o.g. that the served RXs are the $k$ first users, such that $\Uc=\{1,\dots,k\}$. 
			The transmitted signal~$\xv\in \mathbb{C}^k$ is then given by
				\begin{equation}
						\xv= \begin{bmatrix}1\\\bm{0}_{k-1\times 1}\end{bmatrix} s_0 + \sum_{i=1}^k\bT^{\APZF}_{i}\bs_i 				\label{eq:APZF_1}
				\end{equation}
			where
				\begin{itemize}
						\item $\bm{s}_i\in \mathbb{C}^{k-n}$ contains $k-n$ data symbols destined to RX~$i$, which we hence denote as \emph{private}, each one of rate $\alpha^{(n)}\log_2(P)$~bits and power $P^{\alpha^{(n)}}\nicefrac{}{(k(k-n))}$, distributed in an i.i.d. manner. They are precoded with the AP-ZF precoder $\bT^{\APZF}_{i}\in \mathbb{C}^{k\times (k-n)}$ with $n$ active TXs as described in detail in Section~\ref{subse:apzf_expl}. 
						\item $\bm{s}_0\in \mathbb{C}$ is a data symbol destined to be decoded at all users and that we hence denote as \emph{common}, transmitted at rate~$(1-\alpha^{(n)})\log_2(P)$~bits and power $P-P^{\alpha^{(n)}}$. 
				\end{itemize}
			Note that $k-n$ data streams are sent to each RX, but each RX has only one antenna. 
			This overloaded transmission is necessary to take advantage of the $k-n-1$ interference terms generated by the RX's symbols at the other RXs, following the intuition from \cite{MaddahAli2012} that interference can be used as side information. 
			This is detailed in Section~\ref{se:RX_signal}.
			A total of $k(k-n)\alpha^{(n)}\log_2(P)$~bits are sent in one channel use through the private data symbols. Furthermore, an additional data symbol of data rate $(1-\alpha^{(n)})\log_2(P)$~bits is broadcast from TX~$1$. Importantly, we will show that this common data symbol~$s_0$ does not only contain new information bits, but also side information to enable the successful decoding of the private data symbols.


			\subsection{Precoding: AP-ZF with $n$ Active TXs}\label{subse:apzf_expl}
			The proposed precoder can be decoupled such that the precoder for each RX is computed independently of the other RXs up to a power normalization. We describe now the AP-ZF precoder serving a specific RX~$i$ with $n$~Active TXs. 
			This precoder allows us to transmit $k-n$~streams to RX~$i$ while reducing the interference at the $n$ following RXs, i.e., at RXs~$(i+t) \text{ mod } [k]+1$, $\forall t\in \{1,\dots,n\}$. 
			For ease of notation, we omit in the following the modulo operation  as it is clear what an index bigger than~$k$ refers to. The precoder is obtained from distributed precoding at all TXs such that
				\begin{equation}
						\bT^{\APZF}_i=
								\begin{bmatrix}
										\be_1^{\trans}\bT^{\APZF(1)}_i\\\be_2^{\trans}\bT^{\APZF(2)}_i\\ \vdots\\\be_k^{\trans}\bT^{\APZF(k)}_i
								\end{bmatrix}, \label{eq:APZF_2}
				\end{equation}
			where $\be^{\trans}_{\ell}$ refers to the ${\ell}$-th row of the identity matrix $\bI_{k\times k}$, and $\bT^{\APZF(j)}_i$ denotes the AP-ZF precoder computed at TX~$j$. We will therefore consider the design of $\bT^{\APZF(j)}_i$ at TX~$j$. 
				\begin{remark}
						Note that, although TX~$j$ computes the full precoder~$\bT^{\APZF(j)}_i$, only some coefficients are effectively used for the transmission due to the distributed precoding configuration, as made clear in \eqref{eq:APZF_2}. \qed
				\end{remark}
			As a preliminary,  let us define the \emph{Active Channel} $\bH_{\mathrm{A}}\in\mathbb{C}^{n\times n}$ as the channel coefficients from the Active TXs (TX~$1$ to TX~$n$) to the RXs whose received interference is reduced (RX~$i+1$ to RX~$i+n$), i.e.,
				\begin{equation}
						\bH_{\mathrm{A}}\triangleq \bH_{i+1:i+n, 1:n}.	\label{eq:APZF_3}
				\end{equation}
				\begin{figure}
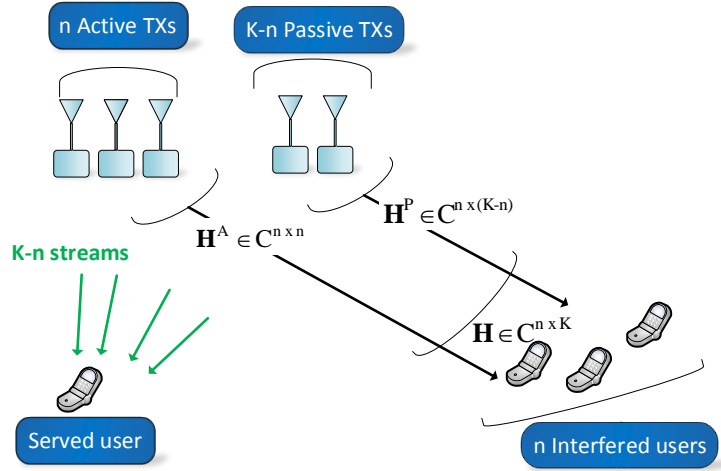
 \centering\vspace{1ex}
						\subfile{Bazco_TIT2018_FIG_setting_APZF_one}%
						\caption{AP-ZF illustration: The number of Active TXs~($n$) determines the number of RXs at which the interference is reduced, whereas the number of Passive TXs ($k-n$) determines the number of data streams that each RX can receive.}%
						\label{channel_description}%
				\end{figure}
			
			\noindent Similarly, the \emph{Passive Channel} $\bH_{\mathrm{P}}\in\mathbb{C}^{n\times (k-n)}$ contains the channel coefficients from the Passive TXs (TX~$n+1$ to TX~$k$) to the RXs with reduced interference (RX~$i+1$ to RX~$i+n$), i.e.,
				\begin{equation}
						\bH_{\mathrm{P}}\triangleq \bH_{i+1:i+n, n+1:k}. \label{eq:APZF_4}
				\end{equation}
			The Passive TXs do not use their instantaneous CSIT. 
			Hence, the precoder applied at the Passive TXs (passive precoder) is an arbitrarily chosen deterministic full-rank matrix denoted as $\lambda^{\APZF}_i\bT^{\mathrm{P}}_i\in\mathbb{C}^{(k-n)\times  (k-n)}$, where $\lambda^{\APZF}_i$ is used to satisfy an average sum power constraint and is detailed further down. 

			On the other hand, every Active TX~$j$, $\forall j\in\{1,\dots,n\}$, computes~$\bT^{\APZF(j)}_i \in\mathbb{C}^{k\times k-n}$ only on the basis of its own available CSIT~$\hat{\bH}^{(j)}$, such that
				\begin{equation}
						\bT^{\APZF(j)}_i=\lambda^{\APZF}_i\begin{bmatrix}\bT^{\mathrm{A}(j)}_i\\\bT^{\mathrm{P}}_i\end{bmatrix}.  			\label{eq:APZF_5}
				\end{equation}
			The precoder applied at the Active TXs (active precoder) is denoted as $\bT^{\mathrm{A}(j)}_i$ and computed as
				\begin{equation}
						\bT^{\mathrm{A}(j)}_i=-\LB \!\!(\hat{\bH}_{\mathrm{A}}^{(j)})^{\He}\hat{\bH}^{(j)}_{\mathrm{A}} +\frac{1}{P}\bI_{n}\!\RB^{-1}(\hat{\bH}^{(j)}_{\mathrm{A}})^{\He} \hat{\bH}^{(j)}_{\mathrm{P}}\bT^{\mathrm{P}}_i. 
						\label{eq:APZF_6}
				\end{equation}
				\begin{remark}
						The design of the active precoder in \eqref{eq:APZF_6} is an extension of the AP-ZF precoder introduced in \cite{dekerret2012_TIT}. 
						Intuitively, the $n$ Active TXs invert the channel to the $n$ chosen RXs so as to cancel the interference generated by the Passive TXs. 
						The number of Passive TXs limits the rank of the transmitted signal, while the number of Active TXs limits the number of users whose received interference is attenuated.\qed
				\end{remark}
			The AP-ZF precoder $\bT^{\APZF}_i \in \mathbb{C}^{k\times k-n}$ practically used in the transmission and defined in \eqref{eq:APZF_2} can be written as
				\begin{equation}
						\bT^{\APZF}_i\triangleq \lambda^{\APZF}_i
								\begin{bmatrix}
										\be_1^{\Trans}\bT^{\mathrm{A}(1)}_i\\
										\vdots\\
										\be_{n}^{\Trans}\bT^{\mathrm{A}(n)}_i\\
										\bT^{\mathrm{P}}_i
								\end{bmatrix}, \label{eq:APZF_7}
				\end{equation} 
			where the normalization coefficient~$\lambda^{\APZF}_i$ is chosen as
				\begin{equation}
						\lambda^{\APZF}_i \triangleq \frac{1}{\sqrt{\E\LSB \left\|\begin{bmatrix}-\LB \bH_{\mathrm{A}}^{\He}\bH_{\mathrm{A}} +\frac{1}{P}\bI_{n}\RB^{-1} \bH_{\mathrm{A}}^{\He} \bH_{\mathrm{P}}\bT^{\mathrm{P}}_i\\\bT^{\mathrm{P}}_i\end{bmatrix}\right\|_{\Fro}^2\RSB}}.	\label{eq:APZF_8}
				\end{equation}
			This normalization constant~$\lambda^{\APZF}_i$ requires only statistical CSI and can hence be computed at every TX, even the passive ones. It ensures that an average sum-power  normalization constraint is satisfied, i.e., that
				\begin{equation}
						\E\LSB \|\bT^{\APZF}_i\|_{\Fro}^2\RSB=1. 						\label{eq:APZF_9}
				\end{equation}  
			The fundamental property of AP-ZF is that it effectively achieves interference reduction at the $n$~RXs up to the worst accuracy across the Active TXs, as stated in the following lemma.
				\begin{lemma}	\label{lemma_fundamental_APZF}
						The AP-ZF precoder with $n$ Active TXs satisfies
							\begin{equation}
									\left\|\bh_{\ell}^{\He}\bT^{\APZF}_i\right\|^2_{\Fro} \dotleq P^{-\alpha^{(n)}},\  \forall \ell \in \{i+1,\ldots,i+n\}.	\label{eq:APZF_10}
							\end{equation}
				\end{lemma}
					\begin{proof}
							The proof of Lemma~\ref{lemma_fundamental_APZF} is given in Appendix~\ref{app:apzf} along with the derivation of other important properties of AP-ZF precoding.
					\end{proof}

			\subsection{Received Signals}\label{se:RX_signal}
			The signal received at RX~$i$ is given by
				\begin{align} 
						\!y_i\! & =\!\underbrace{\bH_{i,1}s_0}_{\doteq P} \;+\; \underbrace{\bh_i^{\He}\bT_{i}^{\APZF}\bm{s}_{i}}_{\doteq P^{\alpha^{(n)}}} \;+\; \underbrace{\bh_i^{\He} \sum_{\mathclap{\ell\in \Ic^{\APZF}_i}} \bT_{\ell}^{\APZF}\bm{s}_{\ell}}_{\doteq P^{\alpha^{(n)}}} \;+\; \underbrace{\bh_i^{\He}  \sum_{\mathclap{\ell\in \Uc\setminus\Ic^{\APZF}_i}} \bT_{\ell}^{\APZF}\bm{s}_{\ell}}_{\doteq P^0},	\label{eq:APZF_11}
				\end{align}
			where the noise term has been neglected for clarity, and where $\Ic^{\APZF}_i$ is defined (omitting the modulo operation) as 
				\begin{align}
						\Ic^{\APZF}_i\triangleq \left \{i+1,\ldots,i+k-n-1\right\}.		\label{eq:APZF_12} 
				\end{align}
			Intuitively, the set~$\Ic^{\APZF}_i$ contains the interfering signals that have not been attenuated towards RX~$i$. The last term in \eqref{eq:APZF_11} scales as $P^{0}$ due to the attenuation produced by the AP-ZF precoding, as shown in Lemma~\ref{lemma_fundamental_APZF}. 
			In Fig.~\ref{fig:ISIT2016_opt_scheme}, we illustrate the received signal at every RX for $k=3$ RXs and $n=1$~Active TX. We can see the improvement with respect to Fig.~\ref{fig:ISIT2016_first_scheme} since the number of significant interference terms is reduced by half thanks to the AP-ZF precoding. 

				\begin{figure*}[t]
					\centering
						\begin{overpic}[width=0.99\textwidth]{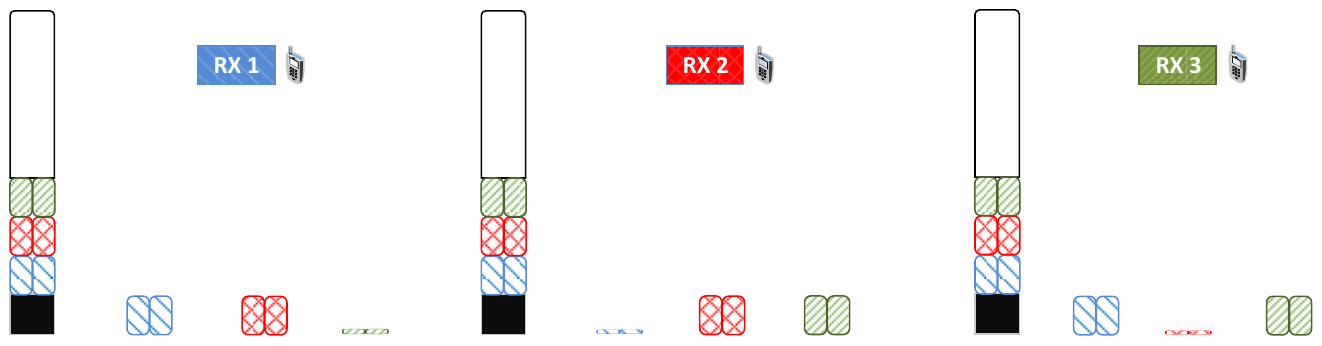}					
								{\put(19,4.5){\parbox{\linewidth}{$\scriptstyle b_1$}}}
								{\put(62,4.5){\parbox{\linewidth}{$\scriptstyle c_1$}}}
								{\put(82,4.5){\parbox{\linewidth}{$\scriptstyle a_1$}}}
													
								{\put(4.5,4.5){\parbox{\linewidth}{$\scriptstyle a_1$}}}
								{\put(4.5,7.5){\parbox{\linewidth}{$\scriptstyle b_1$}}}
								{\put(4.5,10.75){\parbox{\linewidth}{$\scriptstyle c_1$}}}
																	
								{\put(33,12.50){\parbox{\linewidth}{$\scriptstyle s_0\begin{cases}~\\~\vspace{7ex}\\~\\~\\\end{cases}$}}}										
								{\put(0,-2){\parbox{\linewidth}{\resizebox{!}{0.95\height}{$\;\bH_{1,1}s_0  +\, \bh^{\He}_1\bs_1\  + \ \bh^{\He}_1\bs_2\ \color{black!40!white}{\ + \bh^{\He}_1\bs_3}$}}}}
								{\put(35,-2){\parbox{\linewidth}{\resizebox{!}{0.95\height}{$\bH_{2,1}s_0 {\color{black!40!white}{\ \, +\ \,  \bh^{\He}_2\bs_1}} \ + \bh^{\He}_2\bs_2 \; +\;  \bh^{\He}_2\bs_3$}}}}
								{\put(73,-2){\parbox{\linewidth}{\resizebox{!}{0.95\height}{$\!\bH_{3,1}s_0 + \bh^{\He}_3\bs_1 {\color{black!40!white}{\ +\bh^{\He}_3\bs_2}}\, + \bh^{\He}_3\bs_3$}}}}		
								
								{\put(49.75,24.25){\parbox{\linewidth}{\resizebox{!}{.8\height}{$\Pb$\ }}}}				
								{\put(40,24.55){\parbox{\linewidth}{\resizebox{!}{.8\height}{-----------------}}}}				
								{\put(64.5,3.65){\parbox{\linewidth}{\resizebox{!}{.8\height}{---$\Pb^{\alpha^{(1)}}$\!\!\!\!---}}}}
								{\put(64.5,1.65){\parbox{\linewidth}{\resizebox{!}{1.6\height}{$\uparrow$}}}}				
								{\put(69,1.65){\parbox{\linewidth}{\resizebox{!}{1.6\height}{$\uparrow$}}}}								
																				
								\linethickness{0.5mm} \color{blue!85!white}%
										\put(0,0){\polygon(-1.5,-6)(-1.5,26)(30,26)(30,-6)}
								\linethickness{0.5mm} \color{red!85!black}%
										\put(0,0){\polygon(30.5,-6)(30.5,26)(69,26)(69,-6)}
								\linethickness{0.5mm} \color{green!75!black!90!blue}%
										\put(0,0){\polygon(69.5,-6)(69.5,26)(97.5,26)(97.5,-6)}		
																													
						\end{overpic}	
						~\\ \vspace{7ex}  \caption{Illustration of the received signals for the Weak-CSIT regime in the case of $k=3$ Transmitting TXs and $n=1$ Active TX. Due to the AP-ZF precoding, the interference generated is reduced, and thus extra new information can be sent through $s_0$ (white).} \label{fig:ISIT2016_opt_scheme}
				\end{figure*}
			
			\subsection{Decoding}\label{se:weak_CSIT:quantization}
			TX~$1$ uses its local CSIT~$\hat{\bH}^{(1)}$ to estimate the interference terms ${\bh}_i^{\He} \bT^{\APZF}_{\ell}\bm{s}_{\ell}$, $\forall \ell \in \Ic^{\APZF}_i$. 
			Each interference term scales in $P^{\alpha^{(n)}}$. Therefore, each term can be quantized with $\alpha^{(n)}\log_2(P)$~bits such that the quantization noise lies at the noise floor~\cite{Cover2006}. 
			Considering all users, this means that $k(k-n-1)\alpha^{(n)}\log_2(P)$~interference bits have to be transmitted. In order to do so, we use the broadcast data symbol~$s_0$. 
			If the quantity of information to be retransmitted exceeds the data rate of $s_0$, additional broadcast resources will need to be found to enable the successful decoding of the private data symbol. 
			This is the essence of the linear optimization in Theorem~\ref{theo:theorem_achievable} and will be discussed further in Section~\ref{se:Full_scheme}. 
			We assume here that all the interference terms can be transmitted using the common data symbol~$s_0$, and we will verify that it is indeed possible for a given RX~$i$ to decode its $(k-n)\alpha^{(n)}\log_2(P)$~intended bits. 

			By using successive decoding, the data symbol~$s_0$ of rate of $(1-\alpha^{(n)})\log_2(P)$~bits can be decoded with a vanishing probability of error as its effective SNR can be seen in \eqref{eq:APZF_11} to scale as~$P^{1-\alpha^{(n)}}$. Upon decoding~$s_0$, we obtain the estimated interferences $(\hat{\bh}^{(1)}_i)^{\He} \bT^{\APZF(1)}_\ell\bm{s}_\ell$, for $ \ell\in\Ic^{\APZF}_i$, up to the quantization noise at the noise floor. It then holds that
				\begin{align}
						(\hat{\bh}^{(1)}_{i})^{\He} \bT^{\APZF(1)}_{\ell}\bs_{\ell}
								&   = \LB \bh_i^{\He} + \Pb^{-\alpha^{(1)}}(\bm{\delta}^{(1)}_i)^{\He}\RB \bT^{\APZF(1)}_{\ell}\bs_{\ell}\\
								&		= \bh_i^{\He}\bT^{\APZF(1)}_{\ell}\bs_{\ell} + {\Pb^{-\alpha^{(1)}}(\bm{\delta}^{(1)}_i)^{\He}\bT^{\APZF(1)}_{\ell}\bs_{\ell}} \label{eq:APZF_13a}\\
							  &   = \bh_i^{\He}\bT^{\APZF}_{\ell}\bs_{\ell}+{\bh_i^{\He}\LB\bT^{\APZF(1)}_{\ell}-\bT^{\APZF}_{\ell} \RB\bs_{\ell}},\label{eq:APZF_13b}						
				\end{align}
			where \eqref{eq:APZF_13b} is obtained after omitting the second term of \eqref{eq:APZF_13a}, since its power scales as $P^{-\alpha^{(1)}}P^{\alpha^{(1)}} = P^{0}$, i.e., it lies on the noise floor.  It holds that, $\forall \ell\in \{1,\ldots,k\},\forall j\in \{1,\ldots,n\}$, the AP-ZF precoding satisfies the following property (see the proof in Appendix~\ref{app:apzf}):
				\begin{align}
						\|\bT^{\APZF(j)}_{\ell}-\bT^{\APZF}_{\ell}\|^2_{\Fro}\dotleq P^{-\alpha^{(j)}}.		\label{eq:APZF_14a}
				\end{align}
			It follows from \eqref{eq:APZF_14a} that 
				\begin{align}
						{\bh_i^{\He}\LB\bT^{\APZF(1)}_{\ell}-\bT^{\APZF}_{\ell} \RB\bs_{\ell}}\doteq P^0.				\label{eq:APZF_13q}						
				\end{align}
			Once we have subtracted the quantized interference terms, the remaining signal at RX~$i$ up to the noise floor is  
				\begin{align}
						y_i= \bh_i^{\He}\bT^{\APZF}_i\bs_i.					\label{eq:APZF_14b}
				\end{align}
			The key point of our approach is that RX~$i$ also receives  through the broadcast data symbol the interference created by its own intended symbols at the \emph{other RXs}, i.e., the estimated interference terms $(\hat{\bh}_{\ell}^{(1)})^{\He}\bT_i^{\APZF(1)}\bs_i,\ \forall \ell$ such that $i\in \Ic^{\APZF}_\ell$ (note the swap of indexes $i$-$\ell$ with respect to previous expressions). Each of those terms is an independent linear combination of the symbols $\bs_i$, and thus RX~$i$ can form a virtual received vector~$\by_i^{\mathrm{v}}\in\mathbb{C}^{k-n}$ equal to
				\begin{align}
						\by_i^{\mathrm{v}}	\triangleq 
								\begin{bmatrix}
											\bh_i^{\He}\\
											(\hat{\bh}_{{i-1}}^{(1)})^{\He}\\
											\vdots\\
											(\hat{\bh}_{{i-(k-n-1)}}^{(1)})^{\He}
								\end{bmatrix}\bT^{\APZF}_i\bs_i.						\label{eq:APZF_15}
				\end{align}
			Each component of $\by_y^{\mathrm{v}}$ has a SINR scaling in $P^{\alpha^{(n)}}$ and the AP-ZF precoder is of rank~$k-n$ (See Lemma~\ref{lemma_APZF_rank} in Appendix~\ref{app:apzf}) such that RX~$i$ can decode its desired $k-n$~data symbols, each with the rate of $\alpha^{(n)}\log_2(P)$~bits.
				\begin{remark}
						The rank in \eqref{eq:APZF_15} is ensured by the use of the \emph{Passive TXs}. Hence, it is interesting to observe how uninformed TXs prove to be instrumental in the proposed scheme. \qed
				\end{remark}

			\subsection{DoF Analysis}\label{se:DoF}
			The transmission of the data symbols intended to RX~$i$ creates $|\Ic^{\APZF}_i| = k-n-1$ interference terms, which gives in total~$k(k-n-1)\alpha^{(n)}\log_2(P)$~bits that need to be retransmitted. Consequently, we define $\DoF^{\text{Interf(\mhyp)}}_{n,k}$ as the DoF \emph{consumed} in order to transmit these interference terms and which is  given by
				\begin{align}
						\DoF^{\text{Interf}(\mhyp)}_{n,k}\triangleq k(k-n-1)\alpha^{(n)}. \label{eq:dof_interf_n}
				\end{align}  
			In contrast, data symbol $s_0$ carries $(1-\alpha^{(n)})\log_2(P)$~bits, i.e., the DoF of the common data symbol~$\DoF_{n,k}^{\BC}$ is given by
				\begin{align}
						\DoF^{\BC}_{n,k} \triangleq 1-\alpha^{(n)}.					\label{eq:APZF_dof_broadcast}
				\end{align}			
			Finally, considering the $(k-n)\alpha^{(n)}\log_2(P)$ private bits for all $k$~users leads to the private DoF denoted by~$\DoF^{\text{Priv}}_{n,k}$ and defined as
				\begin{align}
						\DoF^{\text{Priv}}_{n,k} \triangleq k(k-n)\alpha^{(n)},				\label{eq:dof_priv_n}
				\end{align}
			which is the DoF obtained from the private data symbols if all the interference is canceled. 
			Putting \eqref{eq:dof_interf_n}, \eqref{eq:APZF_dof_broadcast},  and \eqref{eq:dof_priv_n} together, the total DoF is  
				\begin{align}
						\DoF_{n,k} = \DoF^{\text{Priv}}_{n,k} + \DoF^{\text{\BC}}_{n,k} - \DoF^{\text{Interf}(\mhyp)}_{n,k}
				\end{align} 
			at the condition that $\DoF^{\text{\BC}}_{n,k} - \DoF^{\text{Interf}(\mhyp)}_{n,k} \geq 0$, i.e., that all interference terms could have been retransmitted. If this condition does not hold, the retransmission of the interference is managed through the time-sharing optimization of the different modes as discussed in Section~\ref{se:Full_scheme}. Conversely, the optimal result of Theorem~\ref{theo:weak_CSIT} is achieved if the condition  $\DoF^{\text{\BC}}_{n,k} - \DoF^{\text{Interf}(\mhyp)}_{n,k} \geq 0$ is true for the Transmission Mode $(m,K)$, where $m$ is the number of TXs with $\alpha\expj = \alpha\expo$. By solving the inequality, the maximum value of $\alpha^{\Weak}_m$ is obtained as 
			
				\begin{align}
						\alpha^{\Weak}_m\triangleq \frac{1}{1+K(K-m-1)}. 
				\end{align}
							
		\section{Proof of Theorem~\ref{theo:theorem_achievable}}\label{se:Full_scheme} 
		For a particular Transmission Mode $(n,k)$, let us start by defining $d_{n,k}$ as the difference between the DoF available in the broadcast symbol $s_0$ in~\eqref{eq:APZF_dof_broadcast} and the DoF consumed by the interference to be retransmitted in~\eqref{eq:dof_interf_n}, i.e.,  
					\begin{align}
							d_{n,k} & \triangleq \DoF^{\BC}_{n,k} - \DoF^{\text{Interf}(\mhyp)}_{n,k}\\
										&= 1 - \alpha^{(n)} - k(k-n-1)\alpha^{(n)}.
					\end{align}
		It is not required that each Transmission Mode leads to the transmission of all interference terms, which would imply that every $d_{n,k}$ satisfies $d_{n,k}>0$. 
		In fact, it is only necessary that all interference terms were successfully transmitted at the end of the time sharing between all Transmission Modes. Mathematically, this interference retransmission constraint is written as
			\begin{align}
					\sum_{k=2}^{K}\sum_{n=1}^{k-1}\gamma_{n,k} d_{n,k} \geq 0, 
			\end{align} 
		where $\gamma_{n,k}$ is the time-sharing variable, such that $\gamma_{n,k}\geq 0$ and $\sum_{k=2}^{K}\sum_{n=1}^{k-1}\gamma_{n,k}=1$. 
		By taking into consideration that constraint, the sum DoF can then be rewritten as
			\begin{align}
					\sum_{k=2}^{K}\sum_{n=1}^{k-1}\gamma_{n,k}\DoF_{n,k} &= \sum_{k=2}^{K}\sum_{n=1}^{k-1}\DoF^{\text{Priv}}_{n,k}+\DoF^{\text{\BC}}_{n,k} - \DoF^{\text{Interf}(\mhyp)}_{n,k}\\
								&= \sum_{k=2}^{K}\sum_{n=1}^{k-1}\gamma_{n,k} \LB 1 + (k-1)\alpha^{(n)}\RB.			\label{eq:APZF_dof_total}
			\end{align} 
		The optimal time allocated to each Transmission Mode is obtained by solving the following optimization problem, which concludes the proof:
			\begin{align}
					&\underset{\substack{\gamma_{n,k}}}{\mathrm{maximize}}\quad \sum_{k=2}^{K}\sum_{n=1}^{k-1}\gamma_{n,k} \Big( 1+(k-1)\alpha^{(n)}\Big) \label{eq:Objective_primo_proof}\\
					&\text{subject to\quad}\sum_{k=2}^{K}\sum_{n=1}^{k-1}\gamma_{n,k}=1, \hspace{2ex}\gamma_{n,k}\geq 0, \label{eq:time_sharing_proof}\\
					&\phantom{subject to\quad}\sum_{k=2}^{K}\sum_{n=1}^{k-1}d_{n,k}\gamma_{n,k}   \geq 0. \label{eq:interference_retransmission_condition_proof} 
			\end{align}   
		
		It is important to optimize over both the number of Transmitting TXs and the number of Active TXs. 
		As shown in Fig.~\ref{fig:dof_casesTX}, using $K$ Transmitting TXs can be detrimental depending on the CSI allocation at each TX. 
		The number of Active TXs~($n$) controls how many RXs can have its received interference attenuated and up to which level that interference is reduced, while the number of Transmitting TXs~($k$) controls how many users are served.
		Furthermore, the difference $(k-n)$ impacts on how overloaded the transmission is. 
		For example, setting~$k = n+1$ implies a single stream per RX, and thus no overload. In this case, only $n+1$ RXs are served simultaneously, and no interference is generated. 
		Increasing $k$ while fixing $n$ would increase the number of RXs served and the interference generated, which can be beneficial thanks to the fact that the retransmitted interference is useful for two~RXs.

		\section{Conclusion}
		We have described a novel D-CSIT robust transmission scheme that significantly improves the achieved DoF with respect to state-of-the-art precoding approaches when faced with distributed CSIT. 
		We have first derived an upper bound coined as the \emph{centralized upper bound} and consisting in a genie-aided setting where all the channel estimates are made available at all TXs. 
		We have then shown how this genie-aided upper bound was achieved by the proposed transmission scheme over a range of CSIT configurations, the so-called ``Weak-CSIT'' regime. 
		Surprisingly, this upper bound can even be achieved with the CSI being handed at a single TX, i.e., with all other TXs not having access to CSIT. 
		The proposed robust precoding scheme relies on new methods, such as the transmission of the estimated interference from a single TX before the interference is generated, as well as a modified ZF precoding allowing for an overloaded transmission. 
		These new methods have a strong potential in other wireless configurations with TXs having access to different qualities of CSI. Converting these new innovative transmission schemes into practical transmission schemes in realistic environments is an interesting and ongoing research direction.  
		Such a robust precoding scheme could yield important gains in practice and make advanced precoding schemes more practical. 
		Deriving tighter distributed upper bounds is also an interesting and challenging research problem.


	\appendices

		\section{AP-ZF properties}\label{app:apzf}
		We start by showing some simple but important properties of the AP-ZF precoder. 
		We consider in the following  the precoder for a specific RX~$i$'s data symbols. 
		From symmetry, the precoder satisfies the same properties for any RX, such that we omit hereinafter the RX's sub-index $i$ for clarity. 
		Let us recall that the AP-ZF precoder aims to cancel the interference out at only a subset of $n$ RXs.  
		
			\begin{lemma}\label{lemma_APZF_regularization}
					Let $\bH\in\Cb^{n\times K}$ denote the channel matrix  towards the $n$~RXs whose received interference is reduced. With perfect channel knowledge at all Active TXs, the AP-ZF precoder with $n$ Active TXs and $K-n$ Passive TXs satisfies  
					\begin{equation}
							\bH\bT^{\APZF\star}\xrightarrow[]{P\rightarrow \infty}\bm{0}_{n\times (K-n)},\label{eq:APZF_7_app}
					\end{equation}
					where  $\bT^{\APZF \star}$ denotes the AP-ZF precoder  according to the description in Section~\ref{subse:apzf_expl} but based on perfect CSIT, and it is given as
					\begin{equation}
							\bT^{\APZF\star}\triangleq\lambda^{\APZF}\begin{bmatrix}\bT^{\mathrm{A}\star}\\\bT^{\mathrm{P}}\end{bmatrix}. 
					\end{equation} 
			\end{lemma}  
			\begin{proof}
					Using the well known Resolvent identity \cite[Lemma~$6.1$]{Couillet2011}, we can write that
					\begin{align}
							&\LB \bH_{\mathrm{A}}^{\He}\bH_{\mathrm{A}} +\frac{1}{P}\bI_{n}\RB^{-1}\!\!\!\!\!\! -\LB \bH_{\mathrm{A}}^{\He}\bH_{\mathrm{A}}\RB^{-1}	=-\LB \bH_{\mathrm{A}}^{\He}\bH_{\mathrm{A}}\RB^{-1} \frac{1}{P}\bI_{n} \LB \bH_{\mathrm{A}}^{\He}\bH_{\mathrm{A}} +\frac{1}{P}\bI_{n}\RB^{-1}\!\!\!. \label{eq:APZF_8_app}
					\end{align}
					We can then compute the leaked interference as
					\begin{align}
							\bH \bT^{\APZF\star}&=\ \lambda^{\APZF}\bH_{\mathrm{A}} \bT^{\mathrm{A}\star}+\lambda^{\APZF}\bH_{\mathrm{P}} \bT^{\mathrm{P}}\\
							&\myoverset{(a)}{=}\lambda^{\APZF}\bH_\mathrm{A}  \LB \bH_{\mathrm{A}}^{\He}\bH_{\mathrm{A}}\RB^{-1}   \frac{1}{P}\bI_{n}  \LB \bH_{\mathrm{A}}^{\He}\bH_{\mathrm{A}} +\frac{1}{P}\bI_{n}\RB^{-1}\bH_{\mathrm{A}}^{\He} \bH_{\mathrm{P}}\bT^{\mathrm{P}},	\label{eq:APZF_9_app}
					\end{align}
					where equality $(a)$ follows from inserting \eqref{eq:APZF_8_app} inside the AP-ZF precoder and simplifying. 
					By letting the available power $P$ tend to infinity, the leaked interference tends to zero.
			\end{proof}

			\begin{lemma}\label{lemma_APZF_rank}
					The AP-ZF precoder with $n$ Active TXs and $K-n$ Passive TXs is of rank $K-n$.
			\end{lemma} 
			\begin{proof} 
					The passive precoder was chosen such that~$\bT^{\mathrm{P}}$ is full rank, i.e., of rank $K-n$. For any Active TX~$j$, the precoder~$\bT^{\mathrm{A}(j)}$ is a linear function of~$\bT^{\mathrm{P}}$, such that the effective AP-ZF precoder~$\bT^{\APZF}$ resulting from distributed precoding is exactly of rank $K-n$.
			\end{proof}

		\subsection{Proof of Lemma~\ref{lemma_fundamental_APZF}}
		
			We follow a similar approach as in \cite{dekerret2012_TIT}. 
			We can use once more the resolvent identity \cite[Lemma~$6.1$]{Couillet2011} to approximate the matrix inverse and show that, for any $j \leq n$, 
			\begin{equation}
					\left\| \bT^{\APZF(j)} - \bT^{\APZF\star} \right\|_{\Fro}^2\dotleq P^{-\alpha^{(j)}}.	\label{eq:APZF_11app}
			\end{equation}
			 It then follows that
			\begin{align}
					\left\|\bH \bT^{\APZF} \right\|_{\Fro}^2 &\dotleq \left\|\bH\LB \bT^{\APZF} - \bT^{\APZF\star} \RB\right\|_{\Fro}^2\label{eq:APZF_12_1}\\
					&\dotleq \left\|\bH\right\|_{\Fro}^2 \left\| \bT^{\APZF} - \bT^{\APZF\star} \right\|_{\Fro}^2\\
					&\dotleq  \left\|\bH\right\|_{\Fro}^2   \sum_{j=1}^{n} \big\| \bT^{\APZF(j)} - \bT^{\APZF\star} \big\|_{\Fro}^2\\
					&\dotleq P^{-\min_{j\in \{1,\dots,n\}}\alpha^{(j)}},	\label{eq:APZF_12_4}
			\end{align} 
			where \eqref{eq:APZF_12_1} comes from Lemma~\ref{lemma_APZF_regularization} and \eqref{eq:APZF_12_4} follows from \eqref{eq:APZF_11app}.
				
			\begin{remark}
					The interference attenuation of AP-ZF precoding is only limited by the worst CSI accuracy among the Active TXs, and it does not depend on the CSI accuracy at the Passive TXs.\qed
			\end{remark}
			\FloatBarrier

		\section{Proof of  Lemma~\ref{lem:density_bound}} \label{app:bounded_estimation_proof}
			
		We prove Lemma~\ref{lem:density_bound} (and so Theorem~\ref{theo:centralized_upperbound}) for the broad general case where the estimation noise random variables are mutually independent and they are drawn from continuous distributions with density. 
		We first enunciate some definitions and assumptions that are taken on the random variables and their probability density functions (pdf). 
		Later, we prove the lemma for the $K=2$ users case and, to conclude, we prove the general $K>2$ users case by induction.
		From the independence between different channel coefficients, we can restrict ourselves to an arbitrary link such that we omit the previously used sub-indexes $i,k$.
		\vspace{1ex}

		\subsection{Proof of Lemma~\ref{lem:density_bound}} \label{subse:proof_lem_cond_init}
					
				\subsubsection{Preliminaries}~\vspace{1ex}
			
				We recall that the probability density function of a random variable $\Xc$ is denoted as  $f_{\Xc}(x)$. 
				Next, we introduce three  definitions that are necessary for the proof. 
					\enb
						\item[D1)] \emph{$\varepsilon$-Support:} For any $\varepsilon > 0$, the $\varepsilon$-support of a random variable $\Xc$ is defined as
								\eqm{
									\Sc^{\varepsilon}_{\Xc} \triangleq \{x\mid  f_{\Xc}(x)>\varepsilon\}.
								} 
						\item[D2)] \emph{Bounded Support}: A random variable $\Xc$ is said to have bounded support if $\exists M_{\!\Xc} < \infty$ such that $x \leq \abs{M_{\!\Xc}}$ for any $x \in \Sc^{\varepsilon}_{\Xc}$ and for any $\varepsilon > 0$.
						\item[D3)] \emph{Bounded Probability Density Function}: A random variable $\Xc$ is said to have bounded probability density function if there exists a constant $f^{\max}_{\!\Xc}<\infty$ such that $f_{\!\Xc}(x) \leq f^{\max}_{\!\Xc}$  for any $x$. 
					\ene						
				Furthermore, the following assumptions are adopted:
					\enb
							\item[H1)] $\Hc,\ \Delta^{(j)}$, $\forall j\in\Kc$, are  continuous random variables with density that satisfy D2) and D3). 
							\item[H2)] $P>1$ and $1\geq\alpha^{(1)}\geq\dots\geq\alpha^{(K)}\geq 0$.
							\item[H3)] $\Hc,\ \Delta^{(j)}$,  are  independent of $P$, $\alpha^{(j)}$.							
							\item[H4)] $\Hc_{i,k}$ is independent of $\Hc_{i',k'},\quad \forall (i,k)\neq (i',k')$.
							\item[H5)] $\Delta^{(j)}_{i,k}$ is independent of $\Delta^{(\ell)}_{i',k'},\quad \forall (i,k,j)\neq (i',k',\ell)$.
					\ene
					
				From the above, $\Hc$ and $\Delta^{(j)}$ satisfy Definition~\ref{def:bounded_density}, $\forall j\in\Kc$. 
				The $K$ different estimates of $\Hc$ are $K$ random variables defined as~$\hat{\Hc}^{(j)} \triangleq \Hc + \Pb^{-\alpha^{(j)}}\Delta^{(j)}$. 
				Further, we denote the samples drawn from the aforementioned variables	as $h\sim \Hc$,\quad $\delta^{(j)}\sim \Delta^{(j)}$,\quad $\hat{h}^{(j)}\sim \hat{\Hc}^{(j)}$, such that it follows that $\hat{h}^{(j)}\triangleq h+\Pb^{-\alpha^{(j)}}\delta^{(j)}$. As a refresher, and because we will make extensive use of it, we recall the well-known formula for the pdf of a random variable multiplied by a positive constant. 
					
					\begin{proposition}\label{prop:product_constant} 
							If $\Xc$ is a continuous random variable with probability density function $f_{\Xc}(x)$, then, for $c \!>\! 0$, so is $c\cdot\Xc$ a continuous random variable with probability density function 
							\eqm{	f_{c\Xc}(x) =  \frac{1}{c}f_{\Xc}\LB\frac{x}{c}\RB. \label{eq:constant_pdf0}} 
					\end{proposition}
					
				Furthermore, we present a useful lemma  on the convergence of the estimate variables $\hat{\Hc}^{(j)}$ that we apply during the proof. 
				
					\begin{lemma}\label{lem:convergence_estimation}
						Let $\hat{\Hc}^{(j)}$, $\forall j\in\Kc$, be defined from assumptions H1)-H5), and consider that $\alpha^{(j)}>0$. Then, $f_{\hat{\Hc}^{(j)}}$ converges almost surely to $f_{\Hc}$, i.e.,
									\eqm{
											\lim_{P\rightarrow\infty} f_{\hat{\Hc}^{(j)}}(x) = f_{\Hc}(x).
									}			
					\end{lemma}	
					\begin{corollary}\label{cor:conditional_limit_k}
							Let $\alpha^{(1)}>0$. Then, 
									\eqm{
											\lim_{P\rightarrow\infty} f_{\Delta^{(1)}\mid\hat{\Hc}^{(1)},\dots,\hat{\Hc}^{(K)}}(y|\hat{h}^{(1)},\ \dots,\hat{h}^{(K)}) = f_{\Delta^{(1)}}(y).
									}			
					\end{corollary}							
											
					\begin{proof}
							The proof of both Lemma~\ref{lem:convergence_estimation} and Corollary~\ref{cor:conditional_limit_k} is relegated to  Appendix~\ref{subse:proof_lem_cond}. 
					\end{proof}
				Finally, we recall here the Lebesgue's Dominated Convergence Theorem\cite{Billingsley1995}. 
					\begin{theorem}[\!{\cite[Theorem~16.4]{Billingsley1995}}] \label{theo:dominated_convergence}
						 Let $\{f_n\}$ be a sequence of  functions on the measure space $(\Omega,\Sigma,\mu)$, where $\Omega$ is a non-empty sample space, $\Sigma$ is a $\sigma$-algebra on the space $\Omega$, and $\mu$ a measure on $\LB\Omega,\Sigma\RB$. Suppose that
								\eqm{
										\lim_{n\rightarrow\infty}  f_n(x)   = f(x) 
								}
						almost surely. Further suppose that exists an integrable non-negative function $G$ such that
							$\abs{f_n(x)}\leq G(x),\ \forall n,$
						almost surely. Then $\{f_n\}$ and $f$ are integrable and 
							\eqm{
									\lim_{n\rightarrow\infty} \int_{\Omega} f_n(x)d\mu(x)   = \int_{\Omega} f (x)d\mu(x).
							}
					\end{theorem}
					
					\subsubsection{Proof for the K=2 estimates Case}\label{subsubse:2_est}~\vspace{1ex}
					We can write 
						\eqm{ 
							f_{\Hc| \hat{\Hc}^{(1)}, \hat{\Hc}^{(2)}}(h | \hat{h}^{(1)}, \hat{h}^{(2)})  
									&= \frac{f_{\Hc, \hat{\Hc}^{(1)},\hat{\Hc}^{(2)}}(h, \hat{h}^{(1)}, \hat{h}^{(2)})}{f_{\hat{\Hc}^{(1)},\hat{\Hc}^{(2)}}(\hat{h}^{(1)},\hat{h}^{(2)}) } \label{eq:cond_prob} \\
									&\hspace{00ex} \overset{(a)}{=} \frac{f_{\Hc,\hat{\Hc}^{(1)}}(h, \hat{h}^{(1)})f_{\Pb^{-\alpha^{(2)}}\Delta^{(2)}}(\hat{h}^{(2)}-h)}{f_{\hat{\Hc}^{(1)}}(\hat{h}^{(1)}) f_{\hat{\Hc}^{(2)} \mid  \hat{\Hc}^{(1)}}(\hat{h}^{(2)}\mid  \hat{h}^{(1)})} \label{eq:cond_prob4} \\
									&\hspace{0ex} = {f_{\Hc| \hat{\Hc}^{(1)}}(h|\hat{h}^{(1)})} \frac{f_{\Pb^{-\alpha^{(2)}}\Delta^{(2)}}(\hat{h}^{(2)}-h)}{f_{\hat{\Hc}^{(2)}\mid  \hat{\Hc}^{(1)}}(\hat{h}^{(2)}\mid  \hat{h}^{(1)})}, \label{eq:cond_prob5}
						}				
					where $(a)$ comes from the independence between $\Hc$, $\Delta^{(1)}$, $\Delta^{(2)}$. Furthermore, by applying Bayes' formula we obtain that  
						\begin{align}
								f_{\Hc|\hat{\Hc}^{(1)}}(h\mid \hat{h}^{(1)}) 
											& = \frac{f_{\Pb^{-\alpha^{(1)}}\Delta^{(1)}}( \Pb^{-\alpha^{(1)}}\delta^{(1)}) f_{\Hc}(h)}{f_{\hat{\Hc}^{(1)}}(\hat{h}^{(1)})} \\
											& = \Pb^{\alpha^{(1)}} \frac{ f_{\Delta^{(1)}}(\delta^{(1)}) f_{\Hc}(h)}{f_{\hat{\Hc}^{(1)}}(\hat{h}^{(1)})}, \label{eq:remark_pdf}
						\end{align}	
					where \eqref{eq:remark_pdf} comes from Proposition~\ref{prop:product_constant}. Let us consider separately the cases where $\alpha^{(1)} = 0$ and where $\alpha^{(1)}>0$.
					\paragraph{$\alpha^{(1)} = 0$} 
						In this case, \eqref{eq:remark_pdf} does not depend on $P$, since $P^0 = 1,\ \forall P>0$. 
						From H1), $f_{\Hc}$ and $f_{\Delta^{(1)}}$ are bounded away from $\infty$. 
						Moreover, if $\hat{h}^{(1)}\in\Sc^\varepsilon_{\hat{\Hc}^{(1)}}$, then $ f_{\hat{\Hc}^{(1)}}$ is also lower-bounded by $\varepsilon$. Thus, 
							\begin{align}\label{eq:proof_2_first_final0}
									\max f_{\Hc|\hat{\Hc}^{(1)}}(h\mid \hat{h}^{(1)}) = O\LB{\Pb^{0}}\RB.
							\end{align}	
						
					\paragraph{$\alpha^{(1)} > 0$} 
						From Lemma~\ref{lem:convergence_estimation}, we have that  $f_{\hat{\Hc}^{(1)}}$ converges almost surely (a.s.) to $f_{\Hc}$, and from H1) that $\max f_{\Delta^{(1)}}<\infty$. Thus, from \eqref{eq:remark_pdf} it holds that  
							\begin{align}\label{eq:proof_2_first_final1}
									\max f_{\Hc|\hat{\Hc}^{(1)}}(h\mid \hat{h}^{(1)}) = O\big({\Pb^{\alpha^{(1)}}}\big).
							\end{align}	
						Hence, it follows from~\eqref{eq:proof_2_first_final0}-\eqref{eq:proof_2_first_final1} that, in order to prove Lemma~\ref{lem:density_bound}, i.e., that 
							\begin{align}
									\max f_{\Hc|\hat{\Hc}^{(1)},\hat{\Hc}^{(2)}} = O\big({\Pb^{\alpha^{(1)}}}\big), \label{eq:proof_2_prove}
							\end{align}						
						we need to demonstrate that the limit
							\eqm{
									\lim_{P\rightarrow\infty} \frac{f_{\Hc\mid \hat{\Hc}^{(1)},\hat{\Hc}^{(2)}}(h | \hat{h}^{(1)}, \hat{h}^{(2)})}{f_{\Hc\mid \hat{\Hc}^{(1)}}(h | \hat{h}^{(1)})} \label{eq:cond_prob700}
							}								
					exists and is bounded away from 0 and $\infty$. 
					First, note that from~\eqref{eq:cond_prob5} we have that 
						\eqm{
								&\frac{f_{\Hc\mid \hat{\Hc}^{(1)},\hat{\Hc}^{(2)}}(h | \hat{h}^{(1)}, \hat{h}^{(2)})}{f_{\Hc\mid \hat{\Hc}^{(1)}}(h | \hat{h}^{(1)})} 
										= \frac{f_{\Pb^{-\alpha^{(2)}}\Delta^{(2)}}(\hat{h}^{(2)}-h)}{f_{\hat{\Hc}^{(2)}\mid  \hat{\Hc}^{(1)}}(\hat{h}^{(2)}\mid  \hat{h}^{(1)})}. \label{eq:cond_prob7}
						}
					Let us focus first on the denominator of the right-hand side of \eqref{eq:cond_prob7}.					
					By taking into account again that $\Delta^{(j)}$ is independent of $\Hc$, and that $\hat{\Hc}^{(2)} = \hat{\Hc}^{(1)}\! - \!\Pb^{-\alpha^{(1)}}\!\!\Delta^{(1)}\! +\! \Pb^{-\alpha^{(2)}}\!\!\Delta^{(2)}$,  
					we obtain that 
						\eqm{ 
								f_{\hat{\Hc}^{(2)}\mid  \hat{\Hc}^{(1)}}(\hat{h}^{(2)}\mid  \hat{h}^{(1)})
									& = f_{\hat{\Hc}^{(1)} - \Pb^{-\alpha^{(1)}}\Delta^{(1)} + \Pb^{-\alpha^{(2)}}\Delta^{(2)}\mid  \hat{\Hc}^{(1)}}(\hat{h}^{(1)} +\hat{h}^{(2)}-\hat{h}^{(1)}\mid  \hat{h}^{(1)}) \\				
									& = f_{\Pb^{-\alpha^{(2)}}\Delta^{(2)} - \Pb^{-\alpha^{(1)}}\Delta^{(1)}  \mid  \hat{\Hc}^{(1)}}(\hat{h}^{(2)}-\hat{h}^{(1)} \mid  \hat{h}^{(1)}). \label{eq:eq_var_e}
						}
					Note that $\hat{h}^{(2)}-\hat{h}^{(1)} = \Pb^{-\alpha^{(2)}}\delta^{(2)} - \Pb^{-\alpha^{(1)}}\delta^{(1)}$. 
					From the independence of $\Delta^{(1)}$ and $\Delta^{(2)}$, \eqref{eq:eq_var_e} can be expressed as a convolution such that  
						\eqm{
							f_{\hat{\Hc}^{(2)}\mid  \hat{\Hc}^{(1)}}(\hat{h}^{(2)}\mid  \hat{h}^{(1)}) 
									& = f_{-\Pb^{-\alpha^{(1)}}\Delta^{(1)}\mid\hat{\Hc}^{(1)}}\!\ast\! f_{\Pb^{-\alpha^{(2)}}\Delta^{(2)}}(\hat{h}^{(2)}\!-\!\hat{h}^{(1)}|\hat{h}^{(1)}) \label{eq:cond_prob7a} \\
									& = \int^{\infty}_{-\infty} 
										f_{\Pb^{-\alpha^{(2)}}\Delta^{(2)}}(\hat{h}^{(2)}\!-\!\hat{h}^{(1)}\!-\! x)   f_{-\Pb^{-\alpha^{(1)}}\Delta^{(1)}\mid\hat{\Hc}^{(1)}}(x|\hat{h}^{(1)})\dd x. \label{eq:boundaries}
						}
					We start by applying the change of pdf of Proposition~\ref{prop:product_constant} from $f_{-\Pb^{-\alpha^{(1)}}\Delta^{(1)}}$ to $f_{-\Delta^{(1)}}$  
					such that 
						\eqm{
							f_{\hat{\Hc}^{(2)}\mid  \hat{\Hc}^{(1)}}(\hat{h}^{(2)}\mid  \hat{h}^{(1)}) 
									&\hspace{0ex} = \int^{\infty}_{-\infty} \!\Pb^{\alpha^{(2)}}\! f_{\Delta^{(2)}}\big(\Pb^{\alpha^{(2)}}(\hat{h}^{(2)}-\hat{h}^{(1)}-x)\big)  \Pb^{\alpha^{(1)}}f_{-\Delta^{(1)}\mid\hat{\Hc}^{(1)}}(\Pb^{\alpha^{(1)}}\!x \mid \hat{h}^{(1)})\dd x.
						}
					Changing the integration variable to $y=\Pb^{\alpha^{(1)}}x$ (and thus $\dd x=\Pb^{-\alpha^{(1)}}\dd y$) yields
						\eqm{
							f_{\hat{\Hc}^{(2)}\mid  \hat{\Hc}^{(1)}}(\hat{h}^{(2)}\mid  \hat{h}^{(1)}) 
									& = \Pb^{\alpha^{(2)}}\!\!\int^{\infty}_{-\infty}\!\! f_{\Delta^{(2)}}\big(\Pb^{\alpha^{(2)}}(\hat{h}^{(2)}-\hat{h}^{(1)}-\Pb^{-\alpha^{(1)}}y)\big)  f_{-\Delta^{(1)}\mid\hat{\Hc}^{(1)}}(y|\hat{h}^{(1)})\dd y  \\
									& = \Pb^{\alpha^{(2)}}\int^{\infty}_{-\infty}\upsilon_P(y)\dd y, \label{eq:last_together}
						}
					where 
						\eqm{
									& \upsilon_P(y) \triangleq f_{\Delta^{(2)}}\big(\delta^{(2)}-\Pb^{\alpha^{(2)}-\alpha^{(1)}}(\delta^{(1)}+y)\big) f_{-\Delta^{(1)}\mid\hat{\Hc}^{(1)}}(y|\hat{h}^{(1)}) 	\label{eq:gp_y}							
						}		
					comes from applying  $\hat{h}^{(i)} = h+\Pb^{-\alpha^{(i)}}\delta^{(i)}$. 
					From \eqref{eq:last_together} and by applying again the change of pdf of Proposition~\ref{prop:product_constant} to the numerator, the term $\frac{f_{\Pb^{-\alpha^{(2)}}\Delta^{(2)}}(\hat{h}^{(2)}-h)}{f_{\hat{\Hc}^{(2)}\mid  \hat{\Hc}^{(1)}}(\hat{h}^{(2)}\mid  \hat{h}^{(1)})}$ can be expressed as 
							\eqm{
								\frac{f_{\Pb^{-\alpha^{(2)}}\Delta^{(2)}}(\hat{h}^{(2)}-h)}{f_{\hat{\Hc}^{(2)}\mid  \hat{\Hc}^{(1)}}(\hat{h}^{(2)}\mid  \hat{h}^{(1)})} 
										&= \frac{f_{\Delta^{(2)}}(\delta^{(2)})}{\int^{\infty}_{-\infty} \upsilon_P(y) \dd y}. \label{eq:simple_end_0b}
							}						
					From continuity of $f_{\Delta^{(1)}}$ and $f_{\Delta^{(2)}}$, and Corollary~\ref{cor:conditional_limit_k}, we obtain the limit of $\upsilon_P(y)$ in \eqref{eq:gp_y} when $P\rightarrow\infty$. This limit has two possible expressions depending on the relation between $\alpha^{(1)}$ and $\alpha^{(2)}$. Specifically,  it holds that
							\eqm{
									\lim_{P\rightarrow\infty}  \upsilon_P(y)  = f_{\Delta^{(2)}} (\delta^{(2)}\!-\!\delta^{(1)}-y)  f_{-\Delta^{(1)}}(y)  \label{eq:cases_lim1}
							}										
					if $ \alpha^{(1)} = \alpha^{(2)}$, and that
							\eqm{
									\lim_{P\rightarrow\infty} \upsilon_P(y)  =  f_{\Delta^{(2)}} (\delta^{(2)})  f_{-\Delta^{(1)}}(y)  \label{eq:cases_lim2}
							}	
					if $\alpha^{(1)} > \alpha^{(2)}$. 
					Now we prove separately each of the two possible cases. 
					\setcounter{paragraph}{0}
						\paragraph{$\alpha^{(1)}=\alpha^{(2)}$}
						From the Lebesgue's Dominated Convergence Theorem (Theorem~\ref{theo:dominated_convergence}), D3), and~\eqref{eq:cases_lim1}, the limit exists and it holds that 
							\eqm{
									&\lim_{P\rightarrow\infty} \int^{\infty}_{-\infty} \upsilon_P(y) \dd y  =f_{-\Delta^{(1)}}\ast f_{\Delta^{(2)}}(\delta^{(2)}-\delta^{(1)}). \label{eq:lim_gy_p_a}
							}	
						From~\eqref{eq:simple_end_0b} and \eqref{eq:lim_gy_p_a}, we obtain that 
							\eqm{
									\lim_{P\rightarrow\infty}\frac{f_{\Pb^{-\alpha^{(2)}}\Delta^{(2)}}(\hat{h}^{(2)}- h)}{f_{\hat{\Hc}^{(2)}\mid  \hat{\Hc}^{(1)}}(\hat{h}^{(2)}\mid \hat{h}^{(1)})} 
											&= \frac{f_{\Delta^{(2)}}(\delta^{(2)})}{ f_{-\Delta^{(1)}}\ast f_{\Delta^{(2)}}(\delta^{(2)}-\delta^{(1)})}. \label{eq:simple_end_0z}
							}							
						From D3), we have that, for any $i\in\{1,2\}$, $\exists f^{\max}_{\Delta^{(i)}}<\infty$ such that $f_{\Delta^{(i)}}(x) \leq f^{\max}_{\Delta^{(i)}} \ \ \forall x$. Then, it holds that
							\eqm{
									f_{-\Delta^{(1)}}\ast f_{\Delta^{(2)}}(x)\ \leq\ \max(f^{\max}_{\Delta^{(1)}},f^{\max}_{\Delta^{(2)}}). \label{eq:max_conv}
							}
						Conversely,	let $\mathds{1}$ be the indicator function and let then $\tau$ be
							\eqm{
									\tau \triangleq \int_{-\infty}^{\infty}  \mathds{1}_{x\in\Sc^\varepsilon_{-\Delta^{(1)}}}
											\times \mathds{1}_{(\delta^{(2)}-\delta^{(1)}-x)\in\Sc^\varepsilon_{\Delta^{(2)}}}\dd x.
							}	
						Then, 
							\eqm{
									f_{-\Delta^{(1)}}\ast f_{\Delta^{(2)}}(\delta^{(2)}-\delta^{(1)})\ >\ \epsilon^2\tau  		 \label{eq:min_conv}
							}						
						and $\tau>0$ if	$\delta^{(1)}\in\Sc^\varepsilon_{\Delta^{(1)}}$	and $\delta^{(2)}\in\Sc^\varepsilon_{\Delta^{(2)}}$. 						
						From \eqref{eq:max_conv} and \eqref{eq:min_conv}, \eqref{eq:simple_end_0z} satisfies
							\eqm{
									\frac{\varepsilon}{\max(f^{\max}_{\Delta^{(1)}},f^{\max}_{\Delta^{(2)}})} \	&<\ \frac{f_{\Delta^{(2)}}(\delta^{(2)})}{ f_{-\Delta^{(1)}}\ast f_{\Delta^{(2)}}(\delta^{(2)}-\delta^{(1)})} 
												\ <\ \frac{f^{\max}_{\Delta^{(2)}}}{ \varepsilon^2\tau}. \nonumber
							}	
						This implies \eqref{eq:proof_2_prove} and thus the proof is concluded for the $\alpha^{(2)}=\alpha^{(1)}$ case. \vspace{1ex} 
							
						\paragraph{$\alpha^{(1)}>\alpha^{(2)}$}
						From the Lebesgue's Dominated Convergence Theorem  and~\eqref{eq:cases_lim2}, the limit exists and it holds that
							\eqm{
									\lim_{P\rightarrow\infty}   \upsilon_P(y)	
												& =  \int^{\infty}_{-\infty} f_{\Delta^{(2)}}(\delta^{(2)}) f_{-\Delta^{(1)}}(y)\dd y \label{eq:dominated_applied1} \\
												& =  f_{\Delta^{(2)}}(\delta^{(2)}). \label{eq:dominated_applied2}
							}
						Applying~\eqref{eq:dominated_applied2} in~\eqref{eq:simple_end_0b} yields
							\eqm{
									&\lim_{P\rightarrow\infty}\frac{f_{\Pb^{-\alpha^{(2)}}\Delta^{(2)}}(\hat{h}^{(2)}-h)}{f_{\hat{\Hc}^{(2)}\mid  \hat{\Hc}^{(1)}}(\hat{h}^{(2)}\mid  \hat{h}^{(1)})} = 1,
							}						
							which concludes the proof of Lemma~\ref{lem:density_bound} for the 2-estimate case. \vspace{1ex}								
								
					\subsubsection{Proof for $K>2$ estimates}\label{subsubse:K_est}~\vspace{1ex}
					
					We prove by induction that Lemma~\ref{lem:density_bound} also holds for any number $K$ of estimates. 
					We have proved that it is true for the base cases $K=1$ (trivial) and $K=2$. 
					In the following, we prove the induction step. 
					We denote the set of estimates as $\Gck\triangleq\{\hat{\Hc}^{(1)},\dots,\hat{\Hc}^{(K)}\}$ and, consistently, the set of given values as $\gck\triangleq\{\hat{h}^{(1)},\dots,\hat{h}^{(K)}\}$.
					Let us  assume that Lemma~\ref{lem:density_bound} is verified for a given $K$. 
					We consider ${K+1}$~estimates. Then, from the mutual independence of the estimation noise variables $\Delta^{(j)}$ and Bayes' formula, we obtain that 
						\eqm{ 
								f_{\Hc\mid \Gck,\hat{\Hc}^{(K+1)}}\big(h\mid \gck,\hat{h}^{(K+1)}\big)  = \underbrace{\frac{f_{\Hc, \Gck}\big(h,\gck\big)}{f_{\Gck}\LB\gck\RB}}_{f_{\hat{\Hc}\mid \Gck}(\hat{h}\mid \gck)} \frac{f_{\Pb^{-\alpha^{(K+1)}}\Delta^{(K+1)}}\big(\hat{h}^{(K+1)}-h\big)}{f_{\hat{\Hc}^{(K+1)}|\Gck}\big(\hat{h}^{(K+1)}|\gck\big)}. \label{eq:cond_probaa5}
						} 
					From the induction hypothesis, it holds that
						\begin{align}
								\max f_{\Hc\mid \Gck}(h\mid \gck) = O\big({\Pb^{\alpha^{(1)}}}\big).
						\end{align}						
					Thus, we need to prove that
						\eqm{
								0 < \lim_{P\rightarrow\infty} \frac{f_{\Pb^{-\alpha^{(K+1)}}\Delta^{(K+1)}}\big(\hat{h}^{(K+1)}-h\big)}{f_{\hat{\Hc}^{(K+1)}|\Gck}(\hat{h}^{(K+1)}|\gck)} < \infty. \label{eq:k_users_proof_a} 
						}
					By taking into consideration that 
						\eqm{
								\hat{h}^{(K+1)}-\hat{h}^{(1)} = \Pb^{-\alpha^{(K+1)}}\delta^{(K+1)} - \Pb^{-\alpha^{(1)}}\delta^{(1)},
						}
						the denominator of the expression in \eqref{eq:k_users_proof_a} can be rewritten as  
						\eqm{ 
								f_{\hat{\Hc}^{(K+1)}|\Gck}\big(\hat{h}^{(K+1)}|\gck\big) 
										&  = f_{\hat{\Hc}^{(1)} + \Delta' \mid \Gck}\big(\hat{h}^{(K+1)}-\hat{h}^{(1)}+\hat{h}^{(1)}\mid  \gck\big) \\
										&  = f_{\Delta' \mid  \Gck}\big(\Pb^{-\alpha^{(K+1)}}\delta^{(K+1)} - \Pb^{-\alpha^{(1)}}\delta^{(1)}  \mid  \gck\big), \label{eq:k_users_proof_l}
						}
					where $\Delta' = \Pb^{-\alpha^{(K+1)}}\Delta^{(K+1)} - \Pb^{-\alpha^{(1)}}\Delta^{(1)}$. 
					Hence, we can express $f_{\Delta' \mid  \Gck}$ as a convolution of pdfs and obtain  
						\eqm{ 
							f_{\hat{\Hc}^{(K+1)}|\Gck}\big(\hat{h}^{(K+1)}|\gck\big)  
								&= f_{\Pb^{-\alpha^{(K+1)}}\Delta^{(K+1)}}\ast f_{-\Pb^{-\alpha^{(1)}}\Delta^{(1)}\mid\Gck}\big(\hat{h}^{(K+1)}- \hat{h}^{(1)}|\gck\big).
						}
					We follow the same steps as in \eqref{eq:boundaries}-\eqref{eq:simple_end_0b} to obtain 
						\eqm{
							\frac{f_{\Pb^{-\alpha^{(K+1)}}\Delta^{(K+1)}}\big(\hat{h}^{(K+1)}-h \big)}{f_{\hat{\Hc}^{(K+1)}|\Gck}\big(\hat{h}^{(K+1)}|\gck\big) } 
										& = \frac{f_{\Delta^{(K+1)}}\big(\delta^{(K+1)}\big)}{ \int^{\infty}_{-\infty} f_{\Delta^{(K+1)}}\big(\delta'_y\big) f_{-\Delta^{(1)}\mid\Gck}\big(y|\gck\big)\dd y}, \label{eq:k_end_0b}
						}	
					where we have introduced the notation 
							\eqm{
									\delta'_y \triangleq \delta^{(K+1)}-\Pb^{\alpha^{(K+1)}-\alpha^{(1)}}(\delta^{(1)}+y)
							} 
					for ease of reading. 
					We can see that \eqref{eq:k_end_0b} is equivalent to \eqref{eq:simple_end_0b} with $\Delta^{(K+1)}$ in place of $\Delta^{(2)}$. Then, by following the same derivation as in the 2-estimate case, i.e., using Corollary~\ref{cor:conditional_limit_k} and Lebesgue's Dominated Convergence Theorem, we conclude the induction step. 
					From the base case and the induction step, Lemma~\ref{lem:density_bound} is proven. 
		
				\subsection{Proof of Lemma~\ref{lem:convergence_estimation}} \label{subse:proof_lem_cond}
				Assuming $\alpha^{(i)}>0$, we have that 
					\eqm{
							\lim_{P\rightarrow\infty} f_{\hat{\Hc}^{(i)}}(\hat{h}^{(i)}) 
										&\ = \lim_{P\rightarrow\infty} f_{\Hc +\Pb^{-\alpha^{(i)}}\Delta^{(i)}}(h+\Pb^{-\alpha^{(i)}}\delta^{(i)})\\
										&\ = \lim_{P\rightarrow\infty} \int_{-\infty}^{\infty} f_{\Hc}(h+\Pb^{-\alpha^{(i)}}\delta^{(i)} -x) f_{\Pb^{-\alpha^{(i)}}\Delta^{(i)}}(x)\dd x\\
										&\ = \lim_{P\rightarrow\infty} \int_{-\infty}^{\infty} f_{\Hc}(h+\Pb^{-\alpha^{(i)}}\delta^{(i)} -\Pb^{-\alpha^{(i)}} y) f_{\Delta^{(i)}}(y)\dd y		\label{eq:lem_proof_changevar_int}  \\
										&\ = \int_{-\infty}^{\infty}f_{\Hc}(h) f_{\Delta^{(i)}}(y)\dd y		\label{eq:lem_proof_dominated}	\\
										&\ = f_{\Hc}(h), \label{eq:lem_proof_dominated_last}
					}													
				where~\eqref{eq:lem_proof_changevar_int} comes from applying the relation between $f_{\Pb^{-\alpha_i}\Delta_i}$ and $f_{\Delta_i}$ (see Proposition~\ref{prop:product_constant}), and from the change of integration variable $y=\Pb^{\alpha_i}x$. Finally,~\eqref{eq:lem_proof_dominated} follows from applying Lebesgue's Dominated Convergence Theorem. 
				Therefore, $f_{\hat{\Hc}^{(i)}}$ converges almost surely to~$f_{{\Hc}}$ and hence Lemma~\ref{lem:convergence_estimation} is proven. 
				Then, in order to prove Corollary~\ref{cor:conditional_limit_k}, i.e., that 
					\eqm{
							\lim_{P\rightarrow\infty} f_{\Delta^{(1)}\mid\hat{\Hc}^{(1)},\dots,\hat{\Hc}^{(K)}}(y|\hat{h}^{(1)},\dots,\hat{h}^{(K)}) = f_{\Delta^{(1)}}(y),
					}
				we apply Bayes' formula to obtain
					\eqm{
							f_{\Delta^{(1)}\mid\hat{\Hc}^{(1)},\dots,\hat{\Hc}^{(K)}}(y|\hat{h}^{(1)},\dots,\hat{h}^{(K)}) 
									&  = \frac{f_{\hat{\Hc}^{(1)},\dots,\hat{\Hc}^{(K)}|\Delta^{(1)}}(\hat{h}^{(1)},\dots,\hat{h}^{(K)}|y)f_{\Delta^{(1)}}(y)}{f_{\hat{\Hc}^{(1)},\dots,\hat{\Hc}^{(K)}}(\hat{h}^{(1)},\dots,\hat{h}^{(K)})} \\
									&  = \frac{f_{{\Hc},\hat{\Hc}^{(2)},\dots,\hat{\Hc}^{(K)}}(h,h^{(2)},\dots,h^{(K)})}{f_{\hat{\Hc}^{(1)},\dots,\hat{\Hc}^{(K)}}(h^{(1)},\dots,h^{(K)})}f_{\Delta^{(1)}}(y) \label{eq:lim_is_unity_uu}. 
					}
				The fact that $\alpha^{(1)}>0$ and~\eqref{eq:lem_proof_dominated_last} yield 
					\eqm{
							\lim_{P\rightarrow\infty} \frac{f_{{\Hc},\hat{\Hc}^{(2)},\dots,\hat{\Hc}^{(K)}}(h,h^{(2)},\dots,h^{(K)})}{f_{\hat{\Hc}^{(1)},\dots,\hat{\Hc}^{(K)}}(h^{(1)},\dots,h^{(K)})} = 1. \label{eq:lem_proof_dominated_last_0}
					}
				By taking the limit on \eqref{eq:lim_is_unity_uu} and applying \eqref{eq:lem_proof_dominated_last_0}, it holds  that 
					\eqm{
							\lim_{P\rightarrow\infty} f_{\Delta^{(1)}\mid\hat{\Hc}^{(1)},\dots,\hat{\Hc}^{(K)}}(y|h_1,\dots,h_K) 
								&= f_{\Delta^{(1)}}(y),\label{eq:lim_is_unity2}
					}								
				which concludes the proof.  \qed	
	
			\section{Proof of Corollary~\ref{theorem_two_phases}} \label{app:two_phases}
			
			In this appendix we prove Corollary~\ref{theorem_two_phases}, i.e., that the solution of Theorem~\ref{theo:theorem_achievable} is composed of at most two phases (Transmission Modes). This is equivalent to prove that there exist 
					$ k_1,k_2 \in \Kc,\ n_1 < k_1, \ n_2 < k_2$  
			such that
					\begin{align}
							\gamma_{n_1,k_1} & > 0,\nonumber\\
							\gamma_{n_2,k_2} &\geq 0, \label{eq:optimal_solution}\\
							\gamma_{n,k} &= 0, \qquad \forall (n,k)\neq (n_1,k_1), (n_2,k_2), \nonumber
					\end{align}
			is always an optimal solution of the maximization problem stated in Theorem~\ref{theo:theorem_achievable}. 
			
			Let us denote the number of variables $\gamma_{n,k}$ in the optimization problem as~$D$. From the problem definition, we know that $D = \sum_{i=2}^{K}(i-1) = \frac{K(K-1)}{2}$. 
			In order to simplify the notation, we denote the aforementioned variables with a unique sub-index. 
			Therefore, the variables $\{\gamma_{n,k}\mid k,n\in\Kc,n<k\}$ become $\{\gammabar_i \mid i\in\{1, \dots, D\}\}$. 
			For a given bijective function $f$, the index $i$ is defined as $i\triangleq f(n,k)$.  
			We can choose e.g. $f(n,k) = \frac{(k-1)(k-2)}{2} + n$.  
			
			We recall the optimization problem of Theorem~\ref{theo:theorem_achievable} but, for sake of clarity, we present it in vector notation. 
			For that, let the vector containing the time-sharing variables $\gammabar_i$ be denoted by $\bgamma$, i.e., $\bgamma \triangleq [\gammabar_1,\gammabar_2,\dots,\gammabar_D]$. 
			Similarly, we define the vector $\bF^\alpha$ as the concatenation of the effective DoF of each Transmission Mode (see~\eqref{eq:Objective_primo}), such that  $\bF^\alpha \triangleq [\bF^\alpha_1,\dotsc,\bF^\alpha_D]$, and $\bF^\alpha_i = 1 + (k-1)\alpha^{(n)}$, with $k$, $n$ given by $(n,k) = f^{-1}(i)$. Finally, the vector of terms $d_{n,k} \triangleq  1 - \alpha^{(n)} - k(k-n-1)\alpha^{(n)}$ for the constraint~\eqref{eq:interference_retransmission_condition} 
			is denoted as $\bd$. Hence, the optimization problem of Theorem~\ref{theo:theorem_achievable} can be expressed as			
						\begin{align} 
								\DoF^{\APZF}(\bm{\alpha})=\ \underset{\bgamma}{\mathrm{maximize}} &\quad \bF^\alpha\bgamma\\ 
								\mathrm{subjectto} &\quad \norm{\bgamma}_1=1\label{Time_sharing1_rep}\\
										&\quad\ \bgamma \succeq 0 \label{Time_sharing2_rep}\\
										&\quad\ \bd\bgamma \geq 0,\label{eq:interference_retransmission_condition_rep}
						\end{align} 
			where $\bF^\alpha$ and $\bd$ are constant vectors.
			Let us remind that if a linear programming problem has an optimal solution then it is an extreme point of the feasible set \cite{Luenberger2015}.
			
			The feasible set given by conditions \eqref{Time_sharing1_rep}-\eqref{Time_sharing2_rep}, which is denoted by $\Cc$, is the \emph{probability simplex}\cite{Boyd2004} determined by the unit vectors $e_{1}, \dots, e_{D} \in \Rb^{D}$, and consequently it is a ($D-1$)-dimensional simplex. 
			On the other hand,~\eqref{eq:interference_retransmission_condition_rep} represents a half-space determined by the {vector} hyperplane\cite{Boyd2004} given by $\bd\bgamma = 0$. The vector hyperplane $\bd\bgamma = 0$ is denoted as $\Vc$. 
			We can have different cases depending on how the \emph{probability simplex} $\Cc$ and the half-space determined by the hyperplane  $\Vc$ intersect. Namely:
			\begin{enumerate}
				\item If $\Cc\cap \{\bgamma\mid \bd\bgamma\geq 0\} = \Cc$ (i.e., $\Cc$ is a subset of the half-space), the feasible region is $\Cc$ and the extreme points are the unit vectors $e_{i}$. 
				Therefore, the solution of the problem  uses only a single mode because the only non-zero variable in a unit vector $e_{i}$ is the $i$-th variable.
				\item If  $\Cc\cap \{\bgamma\mid \bd\bgamma\geq 0\} = \emptyset$, there is not feasible solution. However, this is not possible since we have shown that this linear program is always feasible, just choosing $\gamma_{k-1,k}=1$, with $k\in\{2,\dots, K\}$.
				\item If $\Cc\cap \{\bgamma\mid \bd\bgamma\geq 0\} \subset \Cc$, we need to prove that all the extreme points of the resulting set satisfy \eqref{eq:optimal_solution}. 
				Those extreme points will either be the extreme points of $\Cc$ or belong to the intersection between $\Cc$ and~$\Vc$.		
			\end{enumerate}
			From linear algebra, we know that the intersection of an $l$-dimensional and an $m$-dimensional sub-space in the $n$-dimensional space $\Rb^{n}$ has dimension $p_i$ such that 
				\begin{align}
							p_i\geq l+m-n.
				\end{align}
				Thus, in order to obtain the extreme points ($p_i = 0$) of the feasible set, we must obtain the intersection in the space $\Rb^{D}$ between $\Vc$ ($m = D-1$) and the \emph{edges} of $\Cc$, i.e., the $1$-faces (segments) that define $\Cc$. 
			The edges of $\Cc$ are segments that connect two points with a single non-zero variable (the unit vectors), and therefore they belong to a line of only two non-zero variables. 
			Given that the intersection of $\Vc$ with one edge must be a point of the edge, it holds that all the extreme points have at most two non-zero variables, what means that they satisfy \eqref{eq:optimal_solution}. 
			Therefore, Corollary~\ref{theorem_two_phases} is proven. From the previous analysis, it follows that the feasibility set is convex. 
			
			Moreover, as Theorem~\ref{theo:theorem_achievable} is always composed of at most two Transmission Modes, it can be expressed as the following integer linear program:
				\begin{align} 
						\DoF^{\APZF}(\bm{\alpha})=\underset{\substack{k_1,n_1,\\k_2,n_2}}{\mathrm{maximize}} &\; 1+ \rho(k_1-1)\alpha^{(n_1)} + (1-\rho) (k_2-1)\alpha^{(n_2)}  \label{eq:Objective}\\
						\text{subject to } &\; k_1,k_2 \in \{2,\dots,K\},\\
															 &\; n_1\in \{1,\dots,k_1-1\},\\
															 &\; n_2\in \{1,\dots,k_2-1\} \mid d_{n_2,k_2} \geq 0,  \label{eq:time_sharing_integer}
				\end{align} 
			where   $\rho$ is given by  $\rho \triangleq 1$ if $ d_{n_1,k_1} \geq 0 $ and $\rho\triangleq \frac{d_{n_2,k_2}}{d_{n_2,k_2}-d_{n_1,k_1}}$ otherwise.


	\bibliographystyle{IEEEtran}
	\bibliography{IEEEabrv,Literature} 

\end{document}